\documentclass{marticle}
\title{Sequential Sweeps and High Dimensional Expansion}
\author{Vedat Levi Alev\thanks{University of Haifa ---
    \href{mailto:vedatalev@math.haifa.ac.il}{vedatalev@math.haifa.ac.il}}
    \and
    Ori Parzanchevski\thanks{Hebrew University of Jerusalem ---
\href{mailto:parzan@math.huji.ac.il}{parzan@math.huji.ac.il}}}
\date{\today}

\begin{document}

\begin{titlepage}
    \def\thepage{}
    \thispagestyle{empty}

    \maketitle
    \begin{abstract}
        It is well known that the spectral gap of the down-up walk
        over an $n$-partite simplicial complex (also known as Glauber dynamics)
        cannot be better than $O(1/n)$ due to natural obstructions
        such as coboundaries.
        We study an alternative random walk over partite simplicial
        complexes known
        as the sequential sweep or the systematic scan Glauber dynamics: Whereas
        the down-up walk at each step selects a random coordinate and updates it
        based on the remaining coordinates, the sequential sweep goes
        through each
        of the coordinates one by one in a deterministic order and
        applies the same
        update operation. It is natural, thus, to compare $n$-steps
        of the down-up
        walk with a single step of the sequential sweep.
        Interestingly, while the
        spectral gap of the $n$-th power of the down-up walk is still
        bounded from
        above by a constant, under a strong enough local spectral assumption
        (in the sense of Gur, Lifschitz, Liu, STOC 2022) we can show
        that the spectral gap
        of this walk can be arbitrarily close to 1.
        We also study other isoperimetric inequalities for these walks,
        and show that under the assumptions of local entropy contraction
        (related to the considerations of Gur, Lifschitz, Liu),
        these walks satisfy an entropy contraction inequality.
        Concretely, we generalize a result of Lubetzky, Lubotzky, and
        Parzanchevski (Journal of the EMS) about the rapid mixing of
        sequential sweep in Ramanujan complexes to suitable high
        dimensional expanders.
    \end{abstract}

\end{titlepage}
\newpage

\section{Introduction}
\subsection{The Sequential Sweep and the Down-Up Walk}
\footnote{Everything we define in the introductory part of our paper
will be formally defined in \cref{sec:prelim}} Let $\Omega$ be a
finite set of $n$-tuples, i.e.~$\Omega \subset U_1 \times U_2 \times
\cdots \times U_n$, where each $U_i$ is a finite set. Given a tuple $\omega \in
\Omega$, for each $i \in [n]$, the $i$-th update operation $\update_i(\omega)$
returns a uniformly random tuple from $\Omega$ that is conditioned to agree with
$\omega$ everywhere except for the $i$-th coordinate, i.e.~
\[ \Pr_{\omega' \sim \Omega}\sqbr*{ \update_i(\omega) = \omega'}  =
    \Pr_{\omega' \sim \Omega} \sqbr*{\omega' \mid
\omega_j = \omega'_j~~\textrm{for all}~~j \in [n] \setminus i}.\]
The main object of study in this paper will be the random walk called
\underline{sequential sweep}
$\Psq$ on the collection $\Omega$,
otherwise known as the \underline{systematic scan Glauber
dynamics}\footnote{More precisely, $n$ steps of the systematic scan
dynamics is the walk we call sequential sweep here.} on $\Omega$: A
single step of this dynamics moves between tuples by applying the $n$
update operations in sequential fashion,
i.e.~writing $\omega^{(t)}$ for the tuple visited by this dynamics at step $t$,
$\omega^{(t+1)}$ is the result of the following random update operations,
\[ \omega^{(t+1)} = (\update_n \circ \update_{n-1} \circ \cdots \circ
    \update_1)(\omega^{(t)}) = \update_n( \update_{n-1}( \cdots
(\update_1(\omega^{(t)})))).\]
It is easy to see that the stationary distribution for this random walk is the
uniform distribution on $\Omega$.

It is instructive to compare $\Psq$ with the
the \underline{down-up walk} $\Pgd$ on $\Omega$ also known as the
\underline{Glauber dynamics}, which moves between tuples by first sampling a
random coordinate $i \in [n]$ and then applying the operation
$\update_i(\bullet)$. In other words, $\Psq$ can be thought as
simulating $n$ successive steps of
$\Pgd$ in which we condition the coordinate that is updated at the $i$-th step
to be $i$ for all $i \in [n]$.
\subsection{The Boolean Hypercube and Our Motivation}
For the simple example of $\Omega = \set{0, 1}^n$, the down-up walk $\Pgd$
corresponds to the usual lazy random walk on the $n$-dimensional hypercube
$H_n$. It is well known that this random walk needs at least $0.5
n\log n - O(n)$ steps to
come within distance 0.01 of the uniform distribution (cf.~\cite[Proposition
7.14]{LevinPW09}), however the random walk $\Psq$ on the other hand can easily
be seen to mix in a single step which amounts to $n$ update operations. More
generally, it is known that there are natural obstructions preventing the
down-up walk $\Pgd$ from mixing very rapidly. Say, the spectral gap of $\Pgd$
always satisfies $\Gap(\Pgd) \le 2 \cdot n^{-1}$ so long as we have
$|U_i| > 1$ for all
$i$ (cf.~\cite[Proposition 3.3]{AlevL20}).  This implies that $\Gap(\Pgd^n)
\lesssim 1 - e^{-2}$, i.e.~the spectral gap of $n$-successive steps of $\Pgd$
is still bounded from above by a constant. Thus, a
natural question to ask is when $\Psq$ overcomes these
natural barriers and satisfies $\Gap(\Psq) = 1 - o(1)$.

Methods of high-dimensional expansion \cite{KaufmanM17, DinurK17, KaufmanO18,
Oppenheim18,DiksteinDFH18, AnariLOV18, AlevL20} and also the related
frameworks of spectral
independence and entropic independence \cite{AnariLO20, FGYZ20,
    ChenLV20, CGSV20,
    ChenLV21, GuoM21,
AnariJKP22a, AnariJKP22} have been immensely useful in the analysis of the
down-up walk $\Pgd$, leading to many break throughs in the field of Markov
Chains and Monte Carlo Algorithms, \cite{AnariLOV18, AnariLO20,
    ChenLV20, CGSV20,AnariLOVV21,
AlimohammadiASV21, AnariJKP22}. The random walk $\Psq$ is also well-studied by
the Markov chain community, however considerably less is known about it except
in specialized settings,
e.g.~\cite{DiaconisR00,Hayes06, DyerGJ06,DyerGJ08, RobertsR15, FengGW22}. Thus
it is natural to ask whether techniques of high-dimensional
expansion can be of use for the analysis of $\Psq$. We also note that the
analysis of $\Psq$ has an additional layer of complication, caused by this
random walk not being time-reversible\footnote{i.e.~if $\pi$ is the uniform
    distribution satisfying $\pi\Psq = \pi$, we do not have $\Psq =
    \Psq^*$, where
    $\Psq^*$ is the adjoint to $\Psq$ with respect to $\pi$ -- also known as the
time-reversal of $\Psq$.} -- thus a general stream-lined framework for
analyzing these random walks would be very helpful to advance our current
understanding of them.

We show that under a suitable assumption of high-dimensional
expansion (in the sense of \cite{GurLL22}), we can bound the mixing time of the
sequential sweep and show that we indeed have $\Gap(\Psq) = 1 - o(1)$ as the
quality of the high-dimensional expansion gets better and better. In
particular, we show that in Ramanujan complexes of dimension $d$ and
thickness $p$, we have $\Gap(\Psq) \ge 1 - (d-1)/p - O(d^2/p^2)$.

We note that there are several reasons for which $\Psq$ is
interesting to study besides the potentially faster mixing time.
\begin{enumerate}[i]
    \item Since $\Psq$ picks the coordinates to be replaced in
        sequential order, it is more \underline{randomness efficient} than
        the random walk $\Pgd$: A single step of $\Psq$ updates the
        same amount of coordinates as $n$ rounds of $\Pgd$, however
        to pick the replacement indices the latter walk requires $O(n
        \log n)$ extra bits.
    \item Since $\Psq$ goes through the coordinates in sequential
        order, it potentially does a better job at exploiting memory
        localities than the random walk $\Pgd$. This leads to
        algorithms which perform better in practice even when the
        number of updates is the same.
    \item A spectral gap bound for a random walk $\Pii$ means that
        the walk $\Pii$ contracts $\ell_2$-norm of functions
        perpendicular to the constant functions, i.e.~$\norm*{\Pii
        \eff}_{\ell_2} \ll \norm{\eff}_{\ell_2}$ for $\eff \perp
        \one$.\footnote{where $\one$ denotes the constant function}
        Usually, one considers a single step of $\Psq$ as $n$
        separate steps updating the coordinates in sequential order.
        The difficulty here is that the individual update operations
        \underline{might not} contract the $\ell_2$ norm at all. Our
        result concerning the spectral gap shows whereas an
        individual update might not contract the $\ell_2$ norm, when
        one bundles $n$ update operations they \underline{inevitably}
        will. We expect this to be an interesting observation on its
        own right and our techniques to find further applications in
        the study of random walks over high-dimensional expanders.
\end{enumerate}
\subsection{High Dimensional Expansion, Spectral Independence, and
$\ee$-products}
We recall the notion of local-spectral expansion in (partite)
simplicial complexes
introduced in \cite{KaufmanM17, DinurK17, KaufmanO18, Oppenheim18,
DiksteinDFH18, AlevL20}. A (finite) \underline{simplicial complex} of
rank $n$ is a
collection of subsets of size at most $n$ of a finite set $U$, i.e.~$X \subset
\binom{U}{\le n}$, which is closed under containment. $X$ is called
\underline{partite}, if we can partition $U$ into $n$ disjoint subsets $U_1,
\ldots, U_n$ such that any set from $X$ of size $n$ contains a unique element
from each $U_i$ -- the sets $U_i$ are called the sides of the simplicial
complex. In particular, each $\omega \in X$ with $|\omega| = n$ can be
treated as a tuple $(\omega_1, \ldots, \omega_n)$ where $\omega_i \in U_i$. We
will write $X^{(i)} = \set*{ \omega \in X : | \omega|  = i}$. We will also
write $\omega_S = (\omega_s)_{s \in S}$ for all $\omega \in X^{(n)}$.

The local structures in $X$, are
described by its \underline{links}. For a face $\alpha \in X$, we
describe the weighted graph
$G^{(\alpha)} = (V_\alpha, E_\alpha, w_\alpha)$ as follows. We have,
\[ V_\alpha = \set*{ u \in U : \alpha \cup u \in X}~~\textrm{and}~~E_\alpha =
\set*{uv \in \binom{U}{2} : \alpha \cup \set{u, v} \in X}.\]
Every edge $uv$ is assigned a weight based on the number of subsets in $X$ of
size $n$ that contain $\alpha \cup \set{u, v}$, i.e.~
\[w_\alpha(uv) = \Abs*{ \set*{\omega \in X : | \omega | = n, \omega
            \supset \alpha
\cup \set*{u, v} }}.\]
Writing $G^{(\alpha)}$ for the random walk matrix of $G^{(\alpha)}$ also, we
observe that $\lambda_2(G^{(\alpha)})$ measures an amortized notion of the
correlation between the
vertices in $V_\alpha$, i.e.~the impact of the inclusion of any $x \in
V_\alpha$ can make in the containment probability of $y \in \omega$ where
$\omega$ is drawn uniformly at random from all the elements in $X^{(n)}$ which
contain $\alpha$. Formally, we say that the simplicial
complex $X$ is a \underline{$\vec \gamma = (\gamma_0, \ldots,
    \gamma_{n-2})$-local
spectral expander} if for each $\alpha \in X$ with $|\alpha| \le n -2$, we have
$\lambda_2(G^{(\alpha)}) \le \gamma_{|\alpha|}$ We note that when $n = 2$,
this notion is equivalent to the usual notion of a spectral expander,
\cite{HooryLW06, LevinPW09}.

It turns out assuming $\vec\gamma$-expansion one can bound the spectral gap of
the Glauber dynamics $\Pgd$.
\begin{theorem}[\cite{AlevL20}]\label{thm:myn}
    If $X$ is a $\vec \gamma$-local spectral expander, then
    \[ \Gap(\Pgd) \ge \frac{1}{n} \cdot \prod_{i = 0}^{n-2} (1 - \gamma_{i}),\]
    where $\Gap(\Pgd) = 1 - \lambda_2(\Pgd)$. Further, if $|U_i| > 2$
    for all $i = 1,\ldots, n$, then
    $\Gap(\Pgd) \le \frac{2}{n}$.
\end{theorem}
Precursors of this result occured in \cite{KaufmanM17, DinurK17,
KaufmanO18, DiksteinDFH18},

\cref{thm:myn} has found numerous applications in the field of sampling
algorithms through the spectral independence framework,
\cite{AnariLO20, FGYZ20, CGSV20}. We refer the reader to the excellent lecture
notes \cite{StefankovicV23} and the references therein for more information on
this framework.

For our purposes an alternative but related notion of high-dimensional
expansion introduced in \cite{GurLL22, DiksteinD19} will be more
useful to consider.
Whereas, we have said that $\lambda_2(G^{(\alpha)})$ measures an amortized
notion of the correlation between vertices in $V_\alpha$, we will be more
interested in measuring an amortized notion of the correlation between the
vertices in $V_\alpha$ which come from pre-specified parts $U_i$ and $U_j$.
Roughly speaking, given a "tuple" from $X^{(n)}$, we are interested in
measuring the amount of influence that the $i$-th coordinate can exert on the
$j$-th coordinate given that coordinates on the
set $S \subset [n] \setminus \set{i,j}$ agree with some $(x_s)_{s \in
S}$. Formally, we
are interested in the second eigenvalue of the following bipartite  graph
$G_{i, j}^{(\alpha)} = (V_{ij}^{(\alpha)}, E_{ij}^{(\alpha)},
w_{ij}^{(\alpha)})$ for each $\alpha \in X$ such that $|\alpha \cap
U_i| = |\alpha \cap U_j|
= 0$. In words, $\alpha$ can be thought as an assignment to a subset of
coordinates $S \subset [n]$, such that $\set{i,j} \not\in S$. Here,
\[ V_{ij}^{(\alpha)} = \set*{ u \in U_i \cup U_j : \alpha \cup u \in
    X}~~\textrm{and}~~E_{ij}^{(\alpha)} =
\set*{uv  : u \in U_i, v \in U_j, \alpha \cup \set{u, v} \in X}.\]
Similarly, we assign each edge $uv$ a weight between the number of subsets in
$X$ of size $n$ containing $\alpha \cup \set{u, v}$, i.e.~
\[w_{ij}^{(\alpha)}(uv) = \Abs*{ \set*{\omega \in X^{(n)} : | \omega | = n,
\omega_S = \alpha, \omega_i = u, \omega_j = v  }},\]
where we have assumed that $\alpha$ to be a tuple with coordinates in $S \subset
[n]$, $u \in U_i$, and $v \in V_j$. We will call the simplicial complex
\underline{$\vec\ee = (\ee_0, \ldots, \ee_{n-2})$-product} if for all
$\alpha$ and $i,
j$ we have $\lambda_2(G^{(\alpha)}_{i, j}) \le \ee_{|\alpha|}$. We observe,
that the quantity $\lambda_2(G^{(\alpha)}_{i, j}) \le \ee_{|\alpha|}$
quantifies the correlation between the $i$-th and $j$-th coordinate on an
$n$-tuple drawn from $X$ conditioned on agreeing with $\alpha$. We say $X$ is
completely $\ee$-product if it is $\vec \ee = (\ee_i)_{i= 0}^{n-2}$-product
where $\ee_i \le \ee$ for all $i = 0, \ldots, n-2$.

In \cite{GurLL22} this definition was used in proving hypercontractive
estimates for the noise operator over simplicial complexes and in
\cite{DiksteinD19}, it was used in the context of agreement testing. We
recall, that every strong enough local partite expander is also $\vec
\ee$-product \cite[Corollary 7.6]{DiksteinD19}
\begin{theorem}\label{thm:yot}
    Let $X$ be an $n$-partite simplicial complex. If $X$ is a $\vec
    \gamma$-local spectral expander with $\gamma_{n-2} \le \frac{\ee}{(n-2) \ee
    + 1}$, then $X$ is also completely $\ee$-product.
\end{theorem}
And also, make the crude observation,
\begin{obsv}\footnote{We will prove this observation formally.}
    Let $X$ be an $n$-partite simplicial complex. If $X$ is a $\vec
    \gamma$-local spectral expander with $\gamma_{i} \le
    \frac{\ee_i}{n-i-1}$, then $X$ is also $(\ee_0, \ldots, \ee_{n-2})$-product.
\end{obsv}

More generally, we will consider the following bipartite graphs $G_{I,
J}^{(\alpha)} = (V_{I, J}^{(\alpha)}, E_{I, J}^{(\alpha)}, w_{I,
J}^{(\alpha)})$ given $\alpha
\in X$, and disjoint sets $I$ and $J$ such that $|\alpha \cap U_k| = 0$ for all
$k \in I \cup J$. Similarly, we are thinking of $\alpha$ as a partial
assignment to a subset of the coordinates of an $n$-tuple which do not assign
anything to the coordinates in $I$ or $J$. Here,
\[ V_{I, J} = \set*{ \omega_I \in \prod_{i \in I} U_i : \alpha \cup \omega_I
    \in X} \cup
    \set*{ \omega_J \in \prod_{j \in J} U_j : \alpha \cup \omega_J
\in X},\]
and
\[ E_{I, J} = \set*{ \omega_I \omega_J : \omega_I \in \prod_{i \in I},
        \omega_J \in \prod_{j \in I}, \alpha \cup \omega_I \cup \omega_J
\in X},\]
and assign a weight on each edge $\omega_I\omega_J \in E_{I, J}$
based on the number of subsets in $X$ of
size $n$ that contain $\alpha \cup \set{u, v}$, i.e.~
\[w_{I, J}^{(\alpha)}(\omega_I\omega_J) = \Abs*{ \set*{\bar\omega \in X^{(n)}
            :\bar\omega_S = \alpha, \bar\omega_I = \omega_I, \bar\omega_S =
\omega_S}}.\]
where we have assumed $\alpha$ to be a partial assignment to the coordinates in
$S$. Similarly to before, $\lambda_2(G_{I, J}^{(\alpha)})$ is a
measure of correlation
between the coordinates on $I$ and the coordinates in $J$.

In \cite{DiksteinD19, GurLL22}, the following bound is shown,
\begin{theorem}\label{thm:showw}
    If $X$ is an $n$-partite simplicial complex with parts
    $U_1,\ldots, U_n$ which is $\vec \ee$-product
    (where $\ee_i \le \ee$ for all $i = 0, \ldots, n-2$), then for
    all $\alpha$ satisfying $|\alpha \cap U_k| = 0$ for $k \in I \cup
    J$, we have
    \[ \lambda_2(G_{I, J}^{(\alpha)})^2 \le |I||J| \cdot \ee^2.\]
\end{theorem}

\subsection{Results and Techniques}
Let $X$ be an $n$-partite simplicial complex. We write,
\[ \ee^{I \to J} := \ee^{I \to J}_X = \max\set*{ \lambda_2(G_{I, J}^{(\alpha)})
        : \alpha~~\textrm{is an assignment to the coordinates in
        }~~[n]\setminus (I
\cup J)}.\]
In particular, $\ee^{I \to J}$ quantifies the maximum correlation between the
coordinates in $I$ and $J$ from an $n$-tuple drawn from $X$ given any
assignment $\alpha$ to the coordinates in $(I \cup J)^c$.

We have,
\begin{theorem}[Informal Version of \cref{thm:csv}]\label{thm:showa}
    Let $X$ be an $n$-partite simplicial complex. We have,
    \[ \sigma_2(\Psq)^2 \le 1 - \prod_{j = 2}^n (1- (\ee^{[j-1] \to j})^2),\]
    where $\sigma_2(\bullet)$ denotes the second largest singular value of
    $\bullet$.

    In particular,
    \[ \Gap(\Psq) \ge 1 - \sqrt{1 - \prod_{j = 2}^n (1- (\ee^{[j-1] \to
    j})^2)},\]
    where we have written $\Gap(\bullet) = 1 - \sigma_2(\bullet)$.
\end{theorem}
We also give an improved version of \cref{thm:showw},
\begin{theorem}[Simplified Version of \cref{thm:cwadv}]\label{thm:trickleshow}
    Let $(X, \pi)$ be an $n$-partite simplicial complex where $\pi$ is $\vec
    \ee$-product. For any assignment $\alpha$ to $S \subset [n]$ such that $S
    \cap [j] = \varnothing$, we have,
    \begin{equation}\sigma_2(G^{(\alpha)}_{\set j, [j-1]})^2 \le
        1 - \prod_{q = 0}^{j- 2}(1-
        \ee_{|\alpha|+ q}^2)\label{eq:cw}
    \end{equation}
    In particular, we have
    \[ (\ee^{j \to [j-1]})^2 \le 1 - \prod_{q = 0}^{j -
    1}(1 - \ee_{n - j + q}^2)\]
\end{theorem}
In conjunction with \cref{thm:yot} and \cref{thm:showa}, this implies
\begin{corollary}\label{cor:cora}
    If $X$ is completely $\ee$-product, then
    \[ \sigma(\Psq) \le \frac{n \cdot \ee}{\sqrt 2}~~\textrm{and}~~\Gap(\Psq)
    \ge 1 - \frac{n\cdot\ee}{\sqrt 2}.\]
    If $X$ has connected links and is a $\vec \gamma$-local spectral expander
    with $\gamma_{n-2} \le \frac{\ee}{2n}$, then
    \[ \sigma(\Psq) \le \frac{\ee}{\sqrt 2}~~\textrm{and}~~\Gap(\Psq) \ge 1 -
    \frac{\ee}{\sqrt 2}.\]
\end{corollary}

This immediately implies that when $\ee \to 0$ we have
$\Gap(\Psq) \to 1$, i.e.~as the $\ee$ parameters get arbitrarily close to 0,
$\Gap(\Psq)$ gets
arbitrarily close to 1. We observe that while the prior for
\cref{cor:cora} is too strong to
be meaningful in the combinatorial sampling setting,\footnote{Where one
    generally has a $\vec \gamma$-local spectral expander with $\gamma_i =
O(1)/(n-i)$.} for families of
local-spectral expanders of bounded rank, such as \cite{KaufmanO18b,
LubotzkySV05} the assumptions are
much milder. We also note that the spectral mixing time estimates one
obtains through the estimate \cref{cor:cora}, say for the simplicial
complexes of \cite{KaufmanO18b, LubotzkySV05} ,are significantly
sharper than those obtained through \cref{thm:myn} (or those obtained
through even stronger entropic estimates, which we will also discuss below)

Recall that the \underline{mixing time} $T(\delta, \Psq)$ of the
random walk $\Psq$ is the minimal number of steps required for the
distirbution $\mu^{(t)} = \mu^{(0)} \Psq^t$ to be $\delta$ close to
the uniform distribution on $X^{(n)}$ in total variation distance
\emph{for any} distribution $\mu^{(0)}$. We can prove,
\begin{corollary}\label{cor:mixt}
    If $X$ has connected links and is a $\vec\gamma$-local spectral
    expander with $\gamma_{n-2} \le \frac{\ee}{2n}$, we have for any
    $\delta > 0$
    \[ T(\delta, \Psq) \le \left\lceil\frac{\log\parens*{
    \frac{\sqrt{|X^{(n)}|}}{\delta}}}{\log(\sqrt 2/ \ee)}\right\rceil\]
\end{corollary}
We note that even in a case where $\ee = 0$, such as the Boolean
hypercube discussed earlier, the random walk can $\Pgd$ requires
$\Omega(n\log n)$ steps to converge to equilibrium. Thus, in the
setting where the assumptions of \cref{cor:cora} and \cref{cor:mixt}
hold, for $\ee \to 0$, the walk $\Psq$ is \emph{significantly}
faster than $\Pgd^n$ -- which is the natural candidate we should
compare the mixing time of $\Psq$ with as each step of $\Psq$ updates
all $n$ coordinates.

In \cref{sec:ramanujan}, we study the relevant expansion parameters
$\ee^{j \to [j-1]}$ in Ramanujan complexes of thickness $p$. We
summarize our results for these complexes below in the following
theorem.\footnote{Recall that a Ramanujan complex of dimension $d-1$
    is a $d$-partite simplicial complex where the colors are named from
$0$ to $d-1$.}
\begin{theorem}[Summary of
    \cref{thm:ramanujan_local,thm:ramanujan_global,thm:ram-mix}]
    Let $X$ be a Ramanujan complex of thickness $p$ and dimension
    $(d-1)$. We have,
    \begin{equation}
        \parens*{\ee^{[0, j-1] \to j}}^2 = \frac{1}{p+1} \parens*{ 1 -
        \Theta\parens*{p^{-j}} }\qquad\textrm{and} \qquad
        \parens*{\ee^{[0, \ldots,
        d-2] \to d-1}}^2 = \frac{1}{p+1} \cdot \parens*{1-
        \Theta\parens*{p^{-j}}}\label{eq:expansion_show}
    \end{equation}
    for all $j = 1, \ldots, d-2$. Consequently,
    \[ \sigma_2\parens*{\Psq}^2 \le \frac{d+1}{p} -
        O\parens*{\frac{d^2}{p^2}}\qquad\textrm{and}\qquad T(\delta,
    \Psq) = O\parens*{\frac{d^2 \log p + \log(N/\delta)}{\log(p/d)}}\]
    where $N$ is the number of vertices in $X$.
\end{theorem}
We note that the expansion parameters in \cref{eq:expansion_show} are
computed exactly in \cref{sec:ramanujan} and they do not suffer the
co-dimension dependent blowups from \cref{thm:trickleshow}, which
would replace the RHS of the expansion parameters in
\cref{eq:expansion_show} with \emph{roughly} $(j-1)/p$ and
$(d-1)/(p+1)$, compare with \cref{thm:showw}. While our techniques
fall short of recovering the
optimal mixing results of \cite{LubetzyLP20}, we hope our techniques
will find further applications in the study of \emph{spectral}
expansion of Ramanujan complexes. We note that our spectral bound
does not follow from the spectral techniques of \cite{LubetzyLP20} as
this work concerns itself with spectral radii of non-backtracking
operators, which in our context do not neatly translate into a bound
on the singular values.

We note that \cref{thm:showa} is proven using geometric
considerations. For each $I
\subset [n]$, we define the subspace $\Uu_I \subset \RR^{X^{(n)}}$ and the
projection operator $\Quu_I$ as,
\[ \Uu_I = \sspan\set*{ \eff \in \RR^{X^{(n)}} : \eff(\omega) =
        \eff(\bar\omega)~~\textrm{for all}~~\omega,\bar\omega \in
        X^{(n)}~\textrm{such that}~\omega_{[n] \setminus I} = \bar\omega_{[n]
\setminus I}}~~\textrm{and }~~\Quu_I =\Proj(\Uu_I),\]
It is well known that $\Psq = \Quu_{\set 1} \cdots \Quu_{\set n}$, e.g.~see
\cite{Amit91}.
We recall that
the \underline{cosine} between two subspaces $\Uu, \Vv \in
\RR^{X^{(n)}}$ is defined as,
\[ \cos(\Uu, \Vv) = \max\set*{ \inpr{\eff, \gee} : \eff \in \Uu \cap
\Vv^\perp, \gee \in \Vv \cap \Uu^\perp, \norm{\eff} = \norm{\gee} = 1}.\]
Similarly, $\sin^2(\Uu, \Vv) = 1 - \cos^2(\Uu, \Vv)$. Now, we can follow in the
footsteps of \cite{Amit91}, who has proven estimates on the spectral
gap of the sequential sweep
in the Gaussian setting, by appealing to the following theorem:
\begin{theorem}[\cite{SmithSW77}]\label{thm:smiths}
    Let $\Quu_1, \ldots, \Quu_p \in \RR^{\Omega \times \Omega}$ be projection
    operators to sub-spaces $\Uu_1, \ldots, \Uu_p$. Set $\Vv_j = \bigcap_{i =
    1}^j \Uu_i$. Writing $\Quu_\star$ for the projection to $\Vv_p$, we have
    \[ \norm*{\Quu_1 \cdots \Quu_p \eff - \Quu_\star \eff}^2 \le
        \parens*{1 - \prod_{j = 2}^p \sin^2(\Uu_j, \Vv_{j-1}) } \cdot
    \norm*{\eff - \Quu_\star \eff}^2. \]
\end{theorem}
Now, \cref{thm:showa} follows from the following geometric interpretation of
$\ee^{I \to J}$:
\begin{theorem}[Informal Version of \cref{thm:angl} and \cref{prop:ints}]
    Let $X$ be an $n$-partite simplicial complex $I, J \subset [n]$
    such that $I \cap J = \varnothing$. Under
    suitable assumptions,
    \[ \cos(\Uu_I, \Uu_J) \le \ee^{I \to J}~~\textrm{and}~~\bigcap_{i \in I}
    \Uu_i = \Uu_I.\]
\end{theorem}

We also show an analog of \cref{thm:showa} for contraction in relative
entropy: We recall that the KL-divergence between the distributions $\mu$ and
$\pi$ is defined as $\Div(\mu~\|~\pi) = \Exp_{x \sim \pi}\sqbr*{
    \frac{\mu(x)}{\pi(x)} \log
\frac{\mu(x)}{\pi(x)}}$.  The \underline{entropy contraction factor}
$\EC(\Pii)$ of the random walk
operator $\Pii \in \RR^{V \times V}$ with stationary measure $\pi$ is
defined as,
\begin{equation*}
    \EC(\Pii) = 1 - \kappa(\Pii)~~\textrm{where}~~\kappa(\Pii) = \sup\set*{
        \frac{\Div(\mu\Pii~\|~\pi)}{\Div(\mu~\|~\pi)} : \mu~\textrm{is a
    probability distribution on}~V}.
\end{equation*}
We notice that we always have,
\[ \Div(\mu\Pii~\|~\pi) \le (1 - \EC(\Pii)) \cdot \Div(\mu~\|~\pi),\]
for all distributions $\mu$ on $V$.

We define the following entropy contraction parameters,
\begin{equation*}
    \eta^{I \to J} = 1 - \kappa^{I \to J}~~\textrm{where}~~\kappa^{I \to J} =
    \sup\set*{\frac{\Div(\mu\cdot G_{I,
        J}^{(\alpha)}~\|~\pi_{J}^{(\alpha)})}{\Div(\mu~\|~\pi_I^{(\alpha)})}
        : \alpha \in
    X[(I \cup J)^c]~\textrm{and}~\mu \in \triangle_{X_\alpha[I]}}
\end{equation*}
where we have overloaded the notation $G_{I, J}^{(\alpha)}$ to mean the random
walk from the part of $V_{I, J}^{(\alpha)}$ that describes the transition from
assignments of $I$ to assignments of $J$, and $\pi_I^{(\alpha)}$ and
$\pi_J^{(\alpha)}$ are uniform distributions the assignments of $I$ and
assignments of $J$ respectively.

We prove,
\begin{theorem}[Informal Version of \cref{thm:ecc}]\label{t:eca}
    Let $X$ be an $n$-partite simplicial complex. Then,
    \begin{equation}\EC(\Psq) \ge \prod_{j = 1}^{n - 1} \eta^{[j+1,
        n] \to j}.\label{eq:mls-bd}
    \end{equation}
    In particular, writing $\pi$ for the uniform distribution on $X^{(n)}$ we
    have
    \[ \Div(\mu\Psq~\|~\pi) \le (1 - \EC(\Psq)) \cdot \Div(\mu~\|~\pi).\]
\end{theorem}
The proof of \cref{t:eca} follows from an inductive argument. Our key
observation is that conditional on the 1-st coordinate equalling any
given value $x$, the distribution of $\mu \Psq$ on the remaining
coordinates in $[2, \ldots, n]$ can be thought as the distribution
resulting from a sequential sweep in a smaller complex $X_x$ with
faces $(\omega_2, \ldots, \omega_n)$ such that $(x, \omega_2,\ldots,
\omega_n)$ is a facet of $X$. The proof is then completed using the
well-known Garland method \cite{Garland73} and the chain-rule for KL
divergences. There is one interesting \emph{twist}: usually the
Garland method considers the localizations of functions, which need
not be normalized. Applying the Garland method to distributions,
requires understanding how the  \emph{localized/pinned} distributions
interplay with the graphs $G^{(\alpha)}_{I, J}$. To the best of our
knowledge, this is where our proofs differ from existing
\emph{local-to-global} arguments of \cite{AlevL20, CGSV20, GuoM21,
AlevR24}. We hope our proof techniques will find more applications in
the future.

We note that both $\EC(\Psq)$ and $\Gap(\Psq)$ are useful tools for bounding
the \underline{mixing time} $T(\Psq)$ of the random walk $\Psq$,
i.e.~the number of steps one needs to take to
come within small, say 0.01, distance of the uniform distribution. We have,
\begin{align*} T(\Psq) &~\lesssim
    \left\lceil\frac{\log\parens*{|X^{(n)}|}}{\log\parens*{
    1/\sigma_2(\Psq) }}\right\rceil ~\lesssim
    \left\lceil\frac{\log(|X^{(n)}|)}{\Gap(\Psq)}\right\rceil,\\
    T(\Psq) &~\lesssim~
    \left\lceil\frac{\log\log(|X^{(n)}|)}{\log(1/(1-\EC(\Psq)))}\right\rceil\lesssim
    \left\lceil\frac{\log\log\parens*{|X^{(n)}|}}{\EC(\Psq)}\right\rceil
\end{align*}
When $\Gap(\Psq)$ and $\EC(\Psq)$ tend to 1, the first bound in each
inequality is significantly sharper than the second. However, as they
get closer and closer to 0, the latter inequalities are essentially
lossless. In particular, when $\Gap(\Psq)$ and $\EC(\Psq)$ are of the
same order, the
bounds one obtains through $\EC(\Psq)$ are typically much sharper, though
typically this constant is much more troublesome to bound.

\begin{remark}
    We note that our proof technique for \cref{t:eca} is general enough to
    prove contraction estimates for $\Psq$ with respect to general
    $\Phi$-divergences,
    $\Div_\Phi(\mu~\|~\pi) = \Exp_{\omega \sim \pi}\sqbr*{ \Phi\parens*{
    \frac{\mu(\omega)}{\pi(\omega)}}}$ for convex $\Phi$, assuming that
    $G^{(\alpha)}_{[j + 1, n], \set{j}}$ contract $\Phi$-divergences.
    In particular, we could
    give an alternative proof for \cref{thm:showa} by setting
    $\Phi(t) = (t-1)^2$
    and adapting the proof of \cref{t:eca}.\footnote{This would
        correspond to the
        so-called $\chi^2$-divergence. The bound by which a random walk
        contracts the $\chi^2$ divergence is well understood to
    correspond to the (square of the) second singular value.} We will
    refrain from taking this route and
    carrying this out explicitly, as we believe that proving
    \cref{thm:showa} through \cref{thm:smiths} better highlights the geometric
    nature of this result. We refer to
    \cite{Chafai04,BoucheronLM13,Raginsky16} for more information on
    $\Phi$-divergences and the closely related $\Phi$-entropies, and
    inequalities for $\Phi$-entropy contraction called
    \emph{generalized} data-processing inequalities.
\end{remark}
\subsection{Related Work}
A similar proof strategy in proving a spectral gap bound for the sequential
sweep in the Gaussian setting was employed by \cite{Amit91} -- there
he shows that the cosine of the angle between subspaces can be
bounded in terms of the
covariance matrix of the underlying Gaussian distribution, which is
where our analysis
diverges from his. We rely on the Garland method \cite{Garland73} to
bound these values, which has been a staple in high-dimensional expansion
literature, e.g.~\cite{Oppenheim18, Oppenheim18b, KaufmanO18, DiksteinDFH18,
AlevL20, Oppenheim21, Oppenheim23}. In particular, there is a non-trivial
overlap between our techniques and the ones employed in \cite{Oppenheim17,
    Oppenheim18c, Oppenheim21,
Oppenheim23, GrinbaumO22}. Similar considerations about angles between
subspaces have also been investigated in the context of Kazhdan constants
and the property
(T), e.g.~\cite{Kazhdan67, DymaraJ00, Kassabov11}.

Our work is close in spirit to \cite{LubetzyLP20} which study a
non-stuttering version
of the sequential sweep $\Psq$ in Ramanujan complexes, called the
geodesic flow. The main result of \cite{LubetzyLP20} is that for
these geodesic flows on Ramanujan complexes, the total-variation
mixing time is as small as possible up to an additive logarithmic
constant. Namely, a $k$-regular flow on $n$ states achieves
$\epsilon$-mixing after only $\log_k(n)+O(\log\log(n))$ steps. In
fact, this result holds for a larger class of walks on so-called
\emph{Ramanujan digraphs} \cite{Parzan20}. While we cannot recover
this result, our results hold in much greater generality without
assumption of Ramanujanness, and also show that in suitable
high-dimensional expanders the sequential sweep improves upon down-up walks.

As already mentioned a significant amount of work has been done in
investigating the mixing properties of the down-up walk $\Pgd$ from
considerations of high-dimensional expansion. See for example,
\cite{KaufmanM17, DinurK17, KaufmanO18, DiksteinDFH18, AlevL20} for works
investigating local-spectral expansion (and variants) to get a bound on
$\Gap(\Pgd)$. Most of the concrete applications of these results rely on
the spectral independence framework of \cite{AnariLO20, FGYZ20, CGSV20} -- we
refer to the excellent lecture notes \cite{StefankovicV23} and the
references therein for more on
this subject. Similarly, the works of \cite{CryanGM19, GuoM21, ChenLV20,
AnariJKP22a, AnariJKP22} explore bounding $\EC(\Pgd)$ based on considerations
of local entropy contraction. Our bound on $\EC(\Psq)$ is proven using similar
ideas. There has also been some work in connecting these
concepts to classical Markov Chain techniques such as path coupling
\cite{Aldous83, BubleyD97}, see e.g.~\cite{Liu21,BlancaCPCPS22}.

The random walks $G_{I, J}^{(\alpha)}$ that we have considered have also been
studied in the works \cite{DiksteinD19, GurLL22}, with applications in
hypercontractivity and agreement testing. They are a natural generalization of
swap-walks which have been studied in \cite{AlevJT19, DiksteinD19}. Their study
have found applications in approximation algorithms for constraint satisfaction
problems \cite{AlevJT19}, agreement testing \cite{DiksteinD19},
coding theory \cite{AlevJQST19, JeronimoST21}, and
hypercontractivity \cite{BafnaHKL22, BafnaHKL22a, GurLL22}.

Although the sequential sweep $\Psq$ is a very natural random sampling
algorithm considerably less is known about it and the results apply only in
restricted and specialized settings, e.g.~\cite{DiaconisR00, DyerGJ06, Hayes06,
DyerGJ08}. These results show mixing by considering what is known as
the Dobrushin matrix \cite{Dobrushin70} whose entries give
pessimistic estimates on the worst case dependencies between the
coordinates. In \cite{HeEtAl16, RobertsR15} the impact of the order of
the vertices for $\Psq$ was studied and shown that the mixing time of the
walk crucially depends on the order of the vertices. The relation between the
relaxation time of $\Psq$ and $\Pgd$ (with a given order) was also subject of
investigation in \cite{GuoKZ18}. We refer the readers to the papers mentioned
here and the references therein for more on the topic of scan order
and mixing time.

The sequence of updated coordinates in $\Psq$ can be thought as the set
of vertices visited by a (deterministic) walk on the directed cyle
$C_n = ([n], E)$ initiated from the vertex 1. It is now a natural idea
to replace the directed circle with an arbitrary graph $G = ([n], E)$
and picking the updated vertices by means of a random walk on $[n]$.
\cite{AlevR24} shows that if $G$ is a two-sided spectral expander,
the mixing time of this new walk is comparable to the mixing time of
$\Pgd$ under mild assumptions. Their techniques fall short of
establishing a mixing time for the \underline{sequential sweep} in
blackbox fashion since the directed cycle has no two-sided expansion.
However, we observe that their result essentially translates an $O(n
\log n)$ mixing time for $\Pgd$ into an $O(\frac{n \log n}{1 -
\lambda(G)})$ mixing time for this new-walk.  Observing that this
mixing time is roughly the cover time\footnote{The number of steps
    following which the random is expected to have visited all the
vertices after being initialized from an arbitrary vertex.} of the
graph $G$, suggesting that a tighter analysis could perhaps establish
an $O(1)$-mixing time\footnote{i.e.~$O(n)$ update steps} for our
sweep\footnote{where we observe that the cover time of the directed
cycle is $n$, which would correspond to a single sweep}.

Finally, we mention the work \cite{GaitondeM24}, which proves worst
case comparison theorems between the sequential sweep $\Psq$ and the
Glauber dynamics $\Pgd$, which shows that $\Gap(\Psq) =
\Omega(n^{-1}) \cdot \Gap(\Pgd)$. They further give examples of
complexes in which this bound is tight. We note that our focus in
this work are complexes where $\Gap(\Psq)$ is significantly better
than $\Gap(\Psq)$, thus our result can be thought as complementary to
their result.
\subsection{Acknowledgements}
VLA would like to thank Max Hopkins for many insightful discussions on the
topic. OP was supported by the ISF grant
2990/21. VLA was supported by the ERC grant of Alex Lubotzky (European Union's
    Horizon
2020/882751), the ISF grant 2669/21 and ERC grant 834735 of Gil
Kalai, ISF grant 2990/21 of OP, and ISF grant 721/2024 of Uriya
First.  Both authors thank anonymous
referees for providing very useful feedback on an earlier version of
this manuscript.

\section{Preliminaries}\label{sec:prelim}
\subsection{Linear Algebra}
\subsubsection*{Vectors and Inner Products}
Throughout this paper, we will use bold faces for various scalar
functions/vectors,
i.e.~$\eff: V \to \RR$.  For $i \in V$, we will write $\one_i: V \to
\RR$ for the
indicator function of $i$, i.e.~$\one_i(i) = 1$ and $\one_i(j) = 0$ for all $j
\ne i$; and similarly for $S \subseteq V$, we will write $\one_S = \sum_{i \in
S} \one_i$ for the indicator vector of $S$. The notation $\one_V$ will be used
for the vector of all ones, when $V$ is clear from context we will simply write
$\one \in \RR^V$ in place of $\one_V$.

We will reserve $\pi \in \RR^V$ to denote various probability distributions of
interest. In particular, this means $\pi(x) \ge 0$ for all $x \in V$ and
$\sum_{x \in V} \pi(x) = 1$ -- we will adopt the convention that
whenever $y \not\in V$, $\pi(y) = 0$.

Given
$\eff, \gee \in \RR^V$, and a measure $\pi$ on $V$ such that $\pi(x) > 0$
for all $x \in V$, we use the notations $\langle \eff, \gee\rangle_{\pi}$
and $\norm{\eff}_\pi$ to denote the inner-product and the norm with respect to
the distribution $\pi$, i.e.
\begin{equation}\langle \eff, \gee\rangle_{\pi} = \Exp_{x \sim
    \pi}\eff(x)\gee(x) =  \sum_{x \in V}
    \pi(x) \cdot \eff(x)\gee(x)  ~~\textrm{ and }~~ \norm{\eff}_\pi^2
    = \langle \eff,
    \eff\rangle_{\pi}.\label{eq:inpr-defn}
\end{equation}

\subsubsection*{Matrices and Eigenvalues}
In this section, we will recall some results concerning eigenvalues and
eigenvectors of matrices.

We will use serif faces for matrices, i.e.~$\Aye,\Bee \in \RR^{U \times V}$. We
will call a matrix $\Bee \in \RR^{U \times V}$
\underline{row-stochastic} if all entries of
$\Bee$ are non-negative and all rows of $\Bee$ sum up to 1, i.e.~
\[ \textrm{for all } u \in U, v \in V~~\Bee(u, v) \ge 0~~\textrm{and}~~\Bee
\one = \one.\]
The \underline{adjoint} of the operator $\Bee \in \RR^{U \times V}$,
with respect to the
inner-products defined by the measures $\pi_U$ and $\pi_V$ on $U$ and $V$,
is the
operator $\Bee^* \in \RR^{V \times U}$ satisfying
\[ \langle \eff, \Bee
    \gee\rangle_{\pi_U} = \langle \Bee^* \eff,
    \gee\rangle_{\pi_V}~~\textrm{ for all
} \eff \in \RR^U, \gee \in \RR^V.\] If $U = V$ and $\pi_U = \pi_V$, the operator
$\Bee$ is called \underline{self-adjoint} when $\Bee^*= \Bee$.  If $\Bee$ is a
row-stochastic matrix, we will call $\Bee^*$ the
\underline{time-reversal} of $\Bee$ with
respect to $\pi_U,\pi_V$ and
say that $\Bee$ is reversible if $\Bee = \Bee^*$. It is well known that the
operator $\Bee^* \in \RR^{V \times U}$ is uniquely determined by the choice of
$\Bee \in \RR^{U \times V}$ and the
inner-products defined by $\pi_U$ and $\pi_V$ (see
e.g.~\cite[p.~318]{Saloff-Coste97}),
\begin{proposition}\label{prop:adjoint-defn}
    Let $\Bee \in \RR^{U \times V}$ be arbitrary. We write $\Bee^*$ for the
    adjoint operator to $\Bee$ with respect to the inner-products defined by
    the distributions $\pi_U$ and $\pi_V$.
    Then,
    \[ \Bee^*(y, x) = \Bee(x, y) \cdot \frac{\pi_U(x)}{\pi_V(y)}~~\textrm{for
    all}~~x \in U, y \in V.\]
\end{proposition}

We also recall the following standard fact which is an immediate consequence of
\cref{prop:adjoint-defn},
\begin{fact}\label{fac:reversal}
    If $\Bee \in \RR^{U, V}$ is a row-stochastic matrix satisfying $\pi_U \Bee =
    \pi_V$, then the adjoint matrix $\Bee^*$ with respect to $\pi_U, \pi_V$ is
    also row-stochastic and satisfies $\pi_V \Bee^* = \pi_U$.
\end{fact}
Let $\Why \in \RR^{V \times V}$ be a square operator. A vector $\vecc v \in
\CC^V\setminus \set{0}$ is called an \underline{eigenvector} of $\Why$ if there
exists an \underline{eigenvalue} $\lambda \in \CC$
such that, $\Why \vecc v = \lambda \vecc v$. We now recall the spectral
theorem (see, e.g.~\cite{HornJ12})
\begin{theorem}[Spectral Theorem]\label{thm:spectral}
    Let $\Why\in \RR^{V \times V}$ be a self-adjoint operator with respect to
    the inner-product defined by $\pi$. Then, all the eigenvalues
    $\lambda_1(\Why), \ldots, \lambda_{|V|}(\Why)$ are real. Further, $\Why$
    has an orthonormal collection of real eigenvectors $\why_1, \ldots,
    \why_{|V|} \in \RR^V$
    such that,
    \[ \Why = \sum_{i = 1}^{|V|} \lambda_i(\Why) \cdot \why_i \why_i^*
        ~~\textrm{and}~~\langle \why_i, \why_j\rangle_{\pi} =
        0~~\textrm{whenever}~i
    \ne j~~\textrm{and}~~\norm{\why_i}_{\pi} = 1,\]
    where $\why^*_i(v) = \why_i(v) \cdot \pi(v)$ for all $v \in V$.
\end{theorem}

For self-adjoint operators $\Why \in \RR^{V \times V}$ we adopt the
convention of taking
$\lambda_i(\Why)$ to be the $i$-th largest eigenvalue of $\Why$, i.e.~we have
$\lambda_1(\Why) \ge \cdots \ge \lambda_{|V|}(\Why)$. We will also write
$\lambda_{\min}(\Why)$ for the least eigenvalue $\lambda_{|V|}(\Why)$ of
$\Why$. We will call the operator $\Why$ positive semi-definite if $\Why$ is
self-adjoint and we have $\lambda_{\min}(\Why) \ge 0$.

We recall the following fundamental theorem from linear algebra (see,
e.g.~\cite{Bhatia2013}),
\begin{theorem}[Courant-Fischer-Weyl Minimax
    Principle]\label{thm:courant-fischer} Let $\Why \in \RR^{V
    \times V}$ be a self-adjoint operator with respect to the measure
    $\pi$. Then,
    \[ \lambda_j(\Why) = \max_{ \Xx \subseteq \RR^V,\atop \dim \Xx =
        j} \min_{\vecc f \in \Xx,\atop \norm{\eff}_{\pi} = 1}
    \langle \eff, \Why \eff\rangle_{\pi}  \]
    where the minimum runs over all subspaces $\mathcal U$ of $\RR^V$ of
    dimension $k$. Further, the maximizer $\Xx$ is spanned by the top $j$
    eigenvectors of $\Why$, i.e.~there exists $\eff_1,\ldots, \eff_j$ such that
    \[ \Xx = \sspan\set*{\eff_1, \ldots, \eff_j}\]
    such that $\langle \eff_k, \eff_l\rangle_{\pi} = 0$ whenever $k \ne l$,
    $\norm{\eff_k}_{\pi} = 1$ for all $k$, and $\Why \eff_k =
    \lambda_k(\Why) \eff$.
\end{theorem}
We note the following consequence of the Perron-Frobenius Theorem
(cf.~\cite[Theorem 8.4.4]{HornJ12}) and \cref{thm:courant-fischer},
\begin{fact}
    If $\Why \in \RR^{V \times V}$ is a row-stochastic matrix that is
    self-adjoint with respect to the measure $\pi$, then we have
    $\lambda_1(\Why) = 1$ and
    \[ \lambda_2(\Why) = \max\set*{ \inpr*{\eff, \Why \eff}_\pi : \eff \in
    \RR^V, \inpr*{\eff, \one}_\pi = 0, \norm{\eff}_\pi = 1}.\]
\end{fact}

Throughout this paper, we will often deal with operators which are not
self-adjoint, and therefore we will need to deal with singular
values.  Given an operator $\Bee \in \RR^{V \times U}$ we will write
$\sigma_i(\Bee)$ for the $i$-th largest singular value of $\Bee$, i.e.~we have
$\sigma_i(\Bee) = \sqrt{\lambda_i(\Bee^*\Bee)}$.\footnote{Crucially, the
    $i$-th singular value always depends on the choice of the measures which
    define the adjoint matrix. We will mostly deal with row-stochastic
    matrices that satisfy $\pi_U \Bee = \pi_V$ for some choice of distributions
    $\pi_U, \pi_V$. The adjoint matrix $\Bee^*$ is then defined with respect to
these distributions.}
This expression is well-defined since $\Bee^*\Bee$ is
positive-semi definite and therefore $\lambda_i(\Bee^*\Bee) \ge 0$.
If $\Bee \in \RR^{V \times V}$ is
a self-adjoint square operator, by the Spectral \cref{thm:spectral} we can pick
an orthonormal basis of eigenvectors $\set{\vecc w}_{i = 1}^{|V|}$ of
$\Bee$ that satisfy,
\[ \Bee^*\Bee = \sum_{i,j =1}^{|V|} \lambda_i(\Bee)\lambda_j(\Bee) \langle \vecc
    w_i, \vecc w_j\rangle \vecc w_i \vecc w_j^* = \sum_{j=1}^{|V|}
    \lambda_j(\Bee)^2
\vecc w_i \vecc w_i^*,\]
where we have used that $\vecc w_i^* \vecc w_j = \langle \vecc w_i, \vecc
w_j\rangle_\pi = 0$ whenever $i \ne j$ and $\langle \vecc w_i, \vecc
w_i\rangle_\pi =
1$. We summarize a few consequneces of the preceding discussion, which will be
very important when dealing with non-self adjoint row-stochastic matrices $\Bee
\in \RR^{U \times V}$.

\begin{fact}\label{fac:rsig}
    Let $\Bee \in \RR^{U \times V}$ be a row-stochastic matrix satisfying
    $\pi_U \Bee = \pi_V$. Then,
    \begin{enumerate}
        \item $\sigma_1(\Bee) = 1$,
        \item $\sigma_2(\Bee) = \max\set*{ \inpr*{\eff, \Bee
                \gee}_{\pi_U} : \eff \in
                \RR^U, \gee \in \RR^V, \norm{\eff}_{\pi_U} =
                \norm{\gee}_{\pi_V} = 1, \inpr*{\eff, \one_U}_{\pi_U} =
            \inpr*{\gee, \one_V}_{\pi_V} = 0}$.
        \item $\sigma_2(\Bee)
            = \max\set*{
                \frac{\norm*{(\Bee - \one_U\pi_V) \eff}_{\pi_U}}{\norm*{\eff -
                \one_V\pi_V \eff}_{\pi_V}} : \eff \in \RR^V, \norm*{\eff -
            \one_V\pi_V \eff}_{\pi_V} \ne 0}$
    \end{enumerate}
\end{fact}

A matrix $\Quu \in \RR^{V \times V}$ is called an \underline{orthogonal
projection} with
respect to $\pi$ if it
is self adjoint with respect to $\pi$, i.e.~$\Quu = \Quu^*$ and it
satisfies the equation $\Quu^2 =
\Quu$. We will say that $\Quu$ is the orthogonal projection operator to a
subspace $\Uu \subset \RR^V$ if the image $\im(\Quu) = \Quu\RR^V$ equals $\Uu$.
The following result is a well-known consequence of the Spectral
\cref{thm:spectral},
\begin{fact}\label{fac:proj}
    Let $\RR^V$ be a vector space equipped with an inner-product with respect
    to the measure $\pi$. The operator $\Quu \in \RR^{V \times V}$ is
    the orthogonal projection to $\Uu$,
    if and only if there exists an orthonormal basis $\Bb = \set*{ \uuu_1,
    \ldots, \uuu_\ell}$ of $\Uu$ such that,
    \[ \Quu = \sum_{i = 1}^\ell \uuu_i \uuu_i^*.\]
\end{fact}
Finally, we will recall some results concerning the angles between subspaces
and products of projection matrices.
\begin{defn}\label{def:cos}
    Let $\Uu, \Vv \subset \RR^{\Omega}$ equipped with the
    inner-product $\pi$. The cosine of the angle between $\Uu$ and
    $\Vv$ is defined to be,
    \[\cos(\Uu, \Vv) = \max\set*{ \inpr*{\uuu, \vvv}_\pi : \uuu \in
            \Uu \cap (\Uu \cap \Vv)^\perp, \vvv \in (\Uu \cap \Vv)^\perp,
    \norm{\uuu}_\pi = \norm{\vvv}_\pi = 1  }, \]
\end{defn}
\begin{theorem}[\cite{SmithSW77}]\label{thm:prprod}
    Let $\Quu_1, \ldots, \Quu_p \in \RR^{\Omega \times \Omega}$ be projection
    operators to sub-spaces $\Uu_1, \ldots, \Uu_p$. Set $\Vv_j = \bigcap_{i =
    1}^j \Uu_i$. Writing $\Quu_\star$ for the projection to $\Vv_p$, we have
    \[ \norm*{\Quu_1 \cdots \Quu_p \eff - \Quu_\star \eff}_\pi^2 \le
        \parens*{1 - \prod_{j = 2}^p \sin^2(\Uu_j, \Vv_{j-1}) } \cdot
    \norm*{\eff - \Quu_\star \eff}^2_\pi \]
    where $\sin^2(\Uu, \Vv) = 1 - \cos^2(\Uu, \Vv)$.
\end{theorem}
\subsection{Probability Distributions and Relative Entropy}
Let $\Omega$ be a finite set. We will write $\triangle_\Omega$ for the
\underline{probability simplex} with vertices $\Omega$, i.e.~
\[ \triangle_\Omega := \set*{ \mu: \Omega \to \RR_{\ge 0} : \sum_{\omega \in
\Omega} \mu(\omega) = 1}.\]
We will adopt the convention of treating elements of $\mu \in
\triangle_{\Omega}$
as row-vectors. For $\mu \in \triangle_\Omega$, we will write
$\supp(\mu) = \set*{ \omega \in
\Omega: \mu(\omega) > 0}$ for the \underline{support} of $\mu$.
Throughout the paper,
$\Omega$ (or $X$) will always be a set of $n$-tuples for some $n \ge
1$, i.e.~there
exists a finite sequence of sets $\Omega_1, \ldots, \Omega_n$ such that $\Omega
\subset \Omega_1 \times \cdots \times \Omega_n$. Given a set $S \subset [n]$,
we will write $\Omega[S]$ for the \underline{projection} of $\Omega$
to the coordinates in
$S$, i.e.~
\[ \Omega[S] = \set*{ (\omega_s)_{s \in S} : (\omega_1, \ldots, \omega_n) \in
\Omega}.\]
The \underline{$S$-marginal distribution} of $\mu \in \Omega$ is defined by,
\begin{equation}\mu_S( \omega_S ) = \sum_{\bar\omega \in
    \Omega[{S^c}] } \mu(\omega_S \oplus
    \bar\omega)~~\textrm{for all}~~\omega_S \in \Omega[S],\label{eq:m-def}
\end{equation}
where $S^c = [n] \setminus S$.

Let $\omega_S \in \Omega_S$. We will write $\Omega_{\omega_S}$ and
$\mu^{(\omega_S)}$ for the
\underline{$\omega_S$-pinning} of $\Omega$ and $\mu_S$, where
\begin{equation}\Omega_{\omega_S} = \set*{\bar\omega \in \Omega[S^c]
        : \omega_S \oplus
    \bar\omega \in \Omega}~~\textrm{and}~~\mu^{(\omega_S)}(\bar\omega) =
    \frac{\mu(\omega_S \oplus \bar\omega)}{\sum_{\tilde \omega \in \Omega[S^c]}
    \mu(\omega_S \oplus \tilde\omega)},\label{eq:p-def}
\end{equation}
for all $\omega_S \in \Omega[S]$. The following relation between the
$S$-marginal $\mu_S$ and $\omega_S$-pinning
is immediate from the definitions. Let, $S, T \subset [n]$ such that $S \cap T
= \varnothing$,
\begin{equation}
    \mu(\omega_S \oplus \omega_{T}) = \mu_S(\omega_S) \cdot
    \mu_T^{(\omega_S)}(\omega_T)\label{eq:b-rule}
\end{equation}
We note that the statement of \cref{eq:b-rule} is a simple
consequence of the law of
conditional probability (Bayes' rule).

The \underline{Kullback-Leibler divergence} $\Div(\mu~\|~\nu)$ between two
distributions $\mu, \nu \in \triangle_\Omega$ is defined as,
\begin{equation}\label{eq:kl-def}
    \Div(\mu~\|~\nu) = \Exp_{x \sim \nu}\sqbr*{ \frac{\mu(x)}{\nu(x)}
    \log\frac{\mu(x)}{\nu(x)}}\tag{KL divergence}
\end{equation}
where we have adopted the convention $0/0 = 1$ and assumed $\supp(\mu) \subset
\supp(\nu)$.

It is a well-known fact that $\Div(\bullet~\|~\bullet)$ is monotone-decreasing
under the application of a row-stochastic matrix,
\begin{fact}[Data Processing Inequality]\label{fac:dpi}
    Let $\Emm \in \RR^{\Omega \times \Omega'}$ be any row-stochastic
    operator. Then, for any
    pair of distributions $\mu, \nu \in \Delta_\Omega$ we have,
    \[ \Div(\mu\Emm ~\|~\nu \Emm) \le \Div(\mu~\|~\nu).\]
\end{fact}

We also recall the following consequence of the chain-rule for the
\ref{eq:kl-def},
\begin{fact}\label{fac:cr}
    Let $n \ge 0$ finite sets $\Omega_1, \ldots, \Omega_n$ be given. Suppose
    the set $\Omega \subset \Omega_1 \times \cdots \times \Omega_n$ and the
    distribution $\mu \in \triangle_{\Omega}$ are arbitrary. Then, writing $S^c
    = [n] \setminus S$, we have
    \[\Div(\mu~\|~\nu) =  \Div(\mu_S~\|~\nu_S) + \Exp_{\omega_S \sim \mu_S}
    \Div(\mu_{S^c}^{(\omega_S)}~\|~\nu_{S^c}^{(\omega_S)}). \]
\end{fact}

\subsection{Random Walks and Mixing Times}\label{ss:rws}

The \ref{eq:mix} $T(\ee, \Pii)$ of the random walk operator $\Pii \in \RR^{V
\times V}$ is defined
to be the least time step where the distribution of the random walk
is $\ee$-close to the stationary distribution $\pi$ of $\Pii$ in the
$\ell_1$ distance,
i.e.~$\pi^{(t)}$ for the distribution of the random walk after $t$-steps, we
have
\[ T(\ee, \Pii) = \min\set*{ t \in \NN_{\ge 0}: \norm{\pi^{(t)} -
        \pi}_{\ell_1} \le
\ee}.\]
We will write $\Aye(x, \bullet)$ for the $x$-th row of the matrix $\Aye$.
Recalling that the distribution of the random walk $\Pii$ starting from $x$
after $t$ steps is given by $\pi^{(t)} = \Pii^t(x, \bullet)$, we formally define
\begin{equation}
    T(\ee, \Pii) = \min\set*{ t \in \NN_{\ge 0} : \norm{\Pii^t(x, \bullet) -
    \pi^\top}_{\ell_1}\le \ee \textrm{ for all }x \in V}.
    \label{eq:mix}\tag{mixing time}
\end{equation}

We define the \underline{spectral gap} $\Gap(\Pii)$ of the random
walk $\Pii \in \RR^{\Omega
\times \Omega}$ with stationary distribution $\pi$ as,
\begin{equation} \Gap(\Pii) = 1 - \sigma_2(\Pii)\label{eq:gap-def}\tag{spectral
    gap}
\end{equation}
The following mixing time bound is well-known, see for example
\cite[Proposition 1.12]{MontenegroT05}.
\begin{theorem}[Spectral Mixing Time Bound]\label{thm:spec-mix-bd}
    Let $\Pii \in \RR^{V \times V}$ be a random walk matrix $\Pii$
    with stationary distribution $\pi$. One has,
    \[ T(\ee, \Pii) \le \left\lceil \frac{\log\parens*{\parens*{\ee
                    \cdot \sqrt{\min_{x \in V}
        \pi(x)}}^{-1}}}{\log(1/\sigma_2(\Pii))}\right\rceil\le
        \left\lceil\frac{\log\parens*{\parens*{\ee
    \cdot \sqrt{\min_{x \in V} \pi(x)}}^{-1}}}{\Gap(\Pii)}\right\rceil\]
    where $\Gap(\Pii)$ is the \ref{eq:gap-def} of the operator $\Pii$.
\end{theorem}

The \underline{entropy contraction factor} $\EC(\Pii)$ of the random walk
operator $\Pii$ with stationary measure $\pi$ is defined as,
\begin{equation}
    \EC(\Pii) = 1 - \kappa(\Pii)~~\textrm{where}~~\kappa(\Pii) = \sup\set*{
        \frac{\Div(\mu\Pii~\|~\pi)}{\Div(\mu~\|~\pi)} : \mu \in
    \triangle_{\Omega}}\tag{entropy contraction factor}\label{eq:mlsc}
\end{equation}
Then, the following estimate can be used to bound the mixing time, see for
example \cite[Lemma 2.4]{BlancaCPSV21},
\begin{theorem}
    Let $\Pii \in \RR^{V \times V}$ be a random walk operator $\Pii$
    with stationary distribution $\pi$. One has,
    \[ T(\ee, \Pii) \le\frac{C}{\EC(\Pii)} \cdot \log\log \frac{1}{\ee
    \cdot {\min_{x \in V} \pi(x)}},\]
    where $\EC(\Pii)$ is the \ref{eq:mlsc} of $\Pii$ and $C$ is a
    universal
    constant independent of $(P, \pi)$.
\end{theorem}

\subsection{(Partite) Simplicial Complexes}\label{sec:pc}
A \underline{simplicial complex} is a downward closed collection of
subsets of a finite set
$U$. Formally, $X \subset 2^U$ and whenever $\beta \in X$ for all
$\alpha \subset \beta$ we have
$\alpha \in X$. The \underline{rank} of a face $\alpha$ is
$|\alpha|$. Given some $j$, we
will adopt the notation $X^{(j)}$ to refer to the collection faces of $X$ of
rank $j$ and the notation $X^{(\le j)}$ to refer to the collection of
faces of $X$
of rank at most $j$, i.e.~
\[ X^{(j)} := X \cap \binom{U}{j}~~\textrm{and}~~X^{(\le j)} :=
    \bigcup_{i = 0}^j
X^{(i)}.\]
We say $X$ is a \underline{simplicial complex of rank $n$} if the
largest rank of
any face $\alpha \in X$ is $n$. We note
that by definition $X^{(0)} = \set{\varnothing}$.

We say that a simplicial complex $X$ of rank $n$ is pure, if any face
$\alpha \in
X^{(j)}$ for any $j < n$
is contained in another face $\beta \in X^{(n)}$. Equivalently, in a
pure simplicial complex the only inclusion maximal
faces are those of maximal rank. In this article, we will only deal
with pure simplicial complexes.

A rank-$n$ pure simplicial complex $X$ is called $n$-partite if we
can partition $X^{(1)}$ into disjoint sets $X[1], \ldots, X[n]$ such that
\[ \textrm{for all}~\beta \in X^{(n)}~\textrm{and for all}~i=1, \ldots,
    n~~\textrm{we have}~~|\beta
\cap X[i]| = 1. \tag{$n$-partiteness}\]
We will call the sets $X[1], \ldots, X[n]$ the sides of the complex
$X$. Equivalently, every element of a rank-$n$ face $\beta \in X[n]$
comes from a distinct side $X[i]$. We observe that a bipartite graph is a
2-partite simplicial complex.
\begin{remark}
    Our notation differs slightly from the preceding work, the notation $X(j)$
    in the preceding work is often used to denote what we have called
    $X^{(j+1)}$, i.e.~$X(j)$ is the set of $j$-dimensional faces. We
    have avoided this notation to
    (i) prevent potential confusion with $X[j]$ and (ii) to avoid refering to
    an $n$-partite complex as an $(n-1)$-dimensional $n$-partite complex, as
    was done in e.g.~\cite{Oppenheim18b, DiksteinD19}.
\end{remark}
To keep our nomenclature simple, we will simply
refer to a pure $n$-partite simplicial complex of rank $n$ as an
\underline{$n$-partite simplicial
complex}, i.e.~we will not consider $n$-partite complexes which are not pure.

For a face $\alpha \in X$ we introduce the notation, $\typ(\alpha) = \set*{ i
\in [n] : \alpha \cap X[i] \ne \varnothing}$ for the \underline{type}
of the face $\alpha$,
i.e.~the collection of sides of $X$ that $\alpha$ intersects.

For any $i \in [n]$ and $\beta \in X^{(n)}$ we will write $\beta_i
\in X^{(1)}$ for
the unique element of $\beta$ satisfying $\set{\beta_i} = \beta \cap X[i]$. We
will refer to $\beta_i$ as the $i$-th coordinate of $\beta$. We
will also write $\beta_T = \set*{\beta_t : t \in T}$ for all $T \subset [n]$.
We extend this notation to arbitrary faces $\alpha \in X$ and $T
\subset \typ(\alpha)$. In keeping with the view that a face $\alpha \in X$ with
$\typ(\alpha) = \set*{t_1, \ldots, t_k}$ can be represented as a tuple
$(a_{t_1}, \cdots, a_{t_k})$, we will favour the notation $\alpha \oplus
\alpha'$ to denote the union of two faces $\alpha, \alpha' \in X$
with $\typ(\alpha)
\cap \typ(\alpha') = \varnothing$ over the usual notation $\alpha \cup \alpha'$.

We observe that for facets $\beta \in X^{(n)}$, i.e.~faces of maximal rank, we
have $\typ(\beta) = [n]$. Given, $\alpha \in X$ we recall that the
\underline{link} $X_\alpha$ is
defined as,
\[ X_\alpha = \set{ (\beta \setminus \alpha) \in X : \beta \in X,
\beta \supset \alpha}.\]
The following observation is immediate,
\begin{fact}
    Let $X$ be an $n$-partite simplicial complex with
    sides $X[1], \ldots, X[n]$ and $\alpha \in X^{(j)}$ for some $j
    \in [0,n]$. Then,
    the simplicial complex $X_\alpha$ is an $(n - j)$-partite
    simplicial complex with sides $X_\alpha[j] := X[j] \cap
    X_\alpha^{(1)}$ for
    $j \in [n]\setminus \typ(\alpha)$.
\end{fact}

For $T \subset [n]$, we will also introduce the notation $X[T]$ to refer to
all faces of $X$ of type $T$, i.e.~
\[ X[T] = \set*{ \alpha \in X : \typ(\alpha) = T}.\]
Notice,
\[ X^{(n)} = X[1,\ldots, n]~~\textrm{and}~~X^{(j)} = \bigcup_{T \in
\binom{[n]}{j}} X[T].\]
\subsection{Weighted Simplicial Complexes}

A \underline{weighted simplicial} complex $(X, \pi)$ of rank $n$ is a
pure simplicial
complex of rank $n$ where $\pi
:= \pi_n$ is a probablity distribution on $X^{(n)}$ with full
support, i.e.~$\pi \in
\triangle_{X^{(n)}}$ and $\supp(\pi) = X^{(n)}$.
For $j \in [0, n-1]$, we inductively define the probability distributions
$\pi_j: X^{(j)} \to \RR$ as
\begin{equation}
    \pi_j(\alpha) = \frac{1}{j+1} \sum_{\beta \in X[j+1],\atop \beta \subset
    \alpha} \pi_{j+1}(\beta).\label{eq:onestep}
\end{equation}
The distribution $\pi_{j}(\alpha)$ can be thought as
the probability of sampling $\alpha \in X^{(j)}$ by first sampling
some $\beta \sim
\pi_{j+1}$ and then removing one of the elements of $\beta$ uniformly at random.
The following proposition generalizes this observation, and follows from a
simple inductive argument (cf.~\cite[Proposition 2.3.1]{AlevL20})
\begin{proposition}\label{prop:multstep}
    Let $(X, \pi)$ be a simplicial complex of rank $n$. For all $0 \le j
    \le k \le n$ and $\alpha \in X^{(j)}$, one has
    \[ \pi_j(\alpha) = \frac{1}{\binom{k}{j}} \sum_{\beta
    \supset\alpha,\atop \beta \in X^{(k)}} \pi_k(\beta).\]
\end{proposition}
Similarly, given a face $\alpha \in X^{(j)}$, we define the distribution
$\pi^{(\alpha)}$ on $X_\alpha^{(n-j)}$
by conditioning $\pi$ on the containment of $\alpha$, i.e.~for
all $\alpha' \in X_\alpha^{(n- j)}$ we have,
\begin{equation}\pi_{n- j}^{(\alpha)}(\alpha') = \frac{\pi_d(\alpha \cup
    \alpha')}{\sum_{\beta \in X^{(n)},\atop \beta \supset \alpha}
    \pi_n(\beta)} =
    \frac{\pi_d(\alpha \cup \alpha')}{\binom{n}{j} \cdot
    \pi_j(\alpha)},\label{eq:link-def}
\end{equation}
where the last part is due to \cref{prop:multstep}. Analogously, we have
\begin{proposition}\label{prop:linkdef}
    Let $(X, \pi)$ be a simplicial complex of rank $n$. Let $\alpha \in
    X^{(j)}$ and $\tau \in X_\alpha^{(\ell)}$ be faces for some $0
    \le \ell \le j \le
    n$.
    Then,
    \[ \pi_l^{(\alpha)}(\tau) = \frac{\pi_{j+\ell}(\alpha \cup
    \tau)}{\binom{j+\ell}{\ell} \cdot \pi_j(\alpha)}.\]
\end{proposition}

When $(X, \pi)$ is an $n$-partite weighted simplicial complex, for
any $T \subset [n]$, we will define $\pi_T$ as the \underline{marginal
distribution} of
$\pi$ on the coordinates in $T$ according to \cref{eq:m-def}, i.e.~
\[ \pi_T(\omega_T) = \sum_{\bar \omega \in X[T^c]} \pi(\omega_T \oplus
\bar\omega).\]
Similarly, given any $\omega_T \in X[T]$, we define the weighting
$\pi^{(\omega_T)}$ on $(X_{\omega_T}, \pi^{(\omega_T)})$ according to
\cref{eq:p-def}. Thus,
\[ \pi^{(\omega_T)}(\bar\omega) = \Pr_{\omega' \sim \pi}[ \omega'_{T^c} =
    \bar\omega \mid \omega'_T = \bar\omega]  = \frac{\pi(\omega_T \oplus
    \bar\omega)}{\pi_T(\omega_T)}~~\textrm{for all}~~\bar\omega \in
X[T^c].\]

\begin{remark}
    There is some similarity between the notations we have used for
    $\pi_i$ on $X^{(i)}$
    and for the distribution $\pi_{\set i}$ on $X[i]$. We stress
    that these distributions are not identical and that the latter will
    \textbf{always} be marked by
    the curly braces.
\end{remark}
\subsection{Link Graphs and the Influence Matrix}\label{ss:links}
Let $(X, \pi)$ be a pure weighted simplicial complex of rank $n$. For any
face $\alpha \in X^{(\le n -2)}$, the \underline{link graph} $G_\alpha =
(X_\alpha^{(1)}, X_\alpha^{(2)}, \pi^{(\alpha)}_2)$ is the weighted graph that
underlies the simplicial complex $X_\alpha$. We will write $\Emm_\alpha$ for the
(row-stochastic) random walk matrix of $G_\alpha$, i.e.~
\begin{equation}\label{eq:walkdef}
    \Emm_\alpha(x, y) = \frac{\pi_2^{(\alpha)}(x,y)}{2 \cdot
    \pi^{(\alpha)}_1(x)} = \pi_1^{(\alpha \sqcup x)}(y). \tag{random
    walk on~$G_\alpha$}
\end{equation}
Let $(X, \pi)$ be an $n$-partite simplicial
complex. We define the vectors $\pphi_{\alpha, i} \in \RR^{X_\alpha^{(1)}}$ for
all $i \in [n] \setminus \typ(\alpha)$ as follows,
\begin{equation}\pphi_{\alpha,i}(u) =
    \begin{cases} n-|\alpha|-1 & \textrm{if}~v \in X_\alpha[i],\\
        -1 &\textrm{otherwise}.
    \end{cases}\label{eq:phidef}
\end{equation}
We recall the following result of Oppenheim, \cite[Lemma
5.5]{Oppenheim18}. Our statement follows the
presentation  in \cite[Section 4]{Oppenheim18b}.

\begin{theorem}\label{thm:opartite}
    Let $(X, \pi)$ be an n-partite simplicial
    complex and $\alpha \in X^{(\le n-2)}$ a given face of rank at
    most $n-2$. Writing $\Emm_\alpha$ for the
    random-walk matrix of the empty link $G_\alpha$ of $(X, \pi)$ we have,
    \[ \Emm_\alpha \pphi_{\alpha,i} = \frac{-1}{n-|\alpha| - 1} \cdot
        \pphi_{\alpha,i}~~\textrm{for
    all}~i \in [n] \setminus \typ(\alpha),\]
    where the vectors $\pphi_{\alpha,i}$ are defined as in
    \cref{eq:phidef}. The vectors $\pphi_{\alpha,i}$ span the
    entirety of the eigenspace of
    $\Emm_\alpha$ corersponding to $-1/(n -|\alpha| - 1)$ and all the remaining
    negative eigenvalues $\lambda$ of $\Emm_\alpha$ satisfy $|\lambda| <
    1/(n-|\alpha| - 1)$.
\end{theorem}

We note that this means that any link $\Emm_\alpha$ of an $n$-partite
simplicial complex has the trivial eigenvalues $1$ (corresponding to the vector
$\one$) and $-1/(n- |\alpha| - 1)$ (corresponding to the eigenvectors
$\pphi_i$). We also note that,
\[\Tt_\alpha := \sspan\set{\one, \pphi_i : i \in [n] \setminus \typ(\alpha)} =
\sspan\set*{ \one_{X_\alpha[i]} : i \in [n] \setminus \typ(\alpha)}.\]
In particular, $\Tt_\alpha$ is the \underline{span of the trivial
eigenvectors} of
$\Emm_\alpha$ mentioned above.

For $i=0, \ldots, n-2$, we define the parameters
\begin{equation}\gamma_i := \gamma_i(X, \pi) = \max_{\alpha \in X^{(i)}}
    \lambda_2(\Emm_\alpha)\label{eq:linkexp}
\end{equation}
as the largest second eigenvalue of any rank-$i$ face of $(X, \pi)$

Let $\vec c = (c_i)_{i= 0}^{n-2}$ be a given vector. We will call $(X, \pi)$ a
\underline{a one-sided $\vec c$-expander} if we have $\gamma_i
\le c_i$ for all $i = 0, \ldots, n-2$. We also recall the following alternative
characterization of the parameters $c_i$ due to \cite{AnariLO20, CGSV20}
according to the spectral independence framework: Recall that the
\underline{influence matrix} $\Inf_\alpha \in
\RR^{X_\alpha^{(1)} \times X_\alpha^{(1)}}$ is defined as the following matrix,
\[ \Inf_\alpha(x_i, y_j) =
    \begin{cases}
        \Pr_{\omega \sim \pi}[ \omega_j = y_j \mid \omega_i = x_i,
        \omega_{\typ(\alpha)} = \alpha] - \Pr_{\omega
        \sim \pi}[\omega_j = y_j \mid \omega_{\typ(\alpha)} = \alpha]
        & \textrm{ if } i
        \ne j,\\
        0 & \textrm{ if } i = j.
\end{cases}\]
for all $i, j \in [n] \setminus \typ(\alpha)$ and for all $x_i \in X_\alpha[i],
y_j \in X_\alpha[j]$.

We note that with our definitions for all $i \ne j$, we
have $\Inf_\alpha(x_i, y_j) = \pi_{\set j}^{(\alpha \oplus x_i)}(y_j) -
\pi_{\set j}^{(\alpha)}(y_j)$. Let $\vec c = (c_i)_{i= 0}^{n-2}$, the
distribution $\pi$ is called \underline{$\vec
c$-spectrally independent} if we have, $\lambda_{\max}(\Inf_\alpha) \le
c_{|\alpha|}$ for all $\alpha \in X^{(\le n - 2)}$.

The following equivalence result is well known,
\begin{fact}[\cite{CGSV20}]\footnote{This is an easy consequence of the proof
        of Theorem 8 in the referenced paper, though the theorem statement is
    weaker.}
    Let $(X, \pi)$ be an $n$-partite simplicial complex. Writing
    $\Tee_\alpha$ for the
    orthogonal projection to $\Tt_\alpha$, we have
    \[ (n- |\alpha| - 1) \cdot (\Ide - \Tee_\alpha) \Emm_\alpha (\Ide -
    \Tee_\alpha) = \Inf_\alpha,\]
    In particular,
    \[ \frac{\lambda_{\max}(\Inf_\alpha)}{n-|\alpha| - 1} =
    \lambda_2(\Emm_\alpha).\]
\end{fact}
\subsection{Update Operators, Glauber Dynamics, and the Sequential Sweep}
Let $(X, \pi)$ be a pure, $n$-partite simplicial complex of rank $n$. We define
the \underline{update operators} $\Quu_1, \ldots, \Quu_n \in \RR^{X^{(n)} \times
X^{(n)}}$ as follows,
\begin{equation}\label{eq:q-def} \Quu_i(\omega, \omega') =
    \Pr_{\bar\omega \sim \pi}\sqbr*{ \bar\omega =
    \omega' \mid \bar\omega \supset \omega_{[n] \setminus i} } =
    \pi_{\set i}^{(\omega_{[n] \setminus i})}(\omega'_i).
\end{equation}
Thus, $\Quu_{i}$ can be thought as describing the random walk which
makes a transition from $\omega \in X^{(n)}$ to a random state $\omega' \in
X^{(n)}$
by only re-sampling the $i$-th coordinate from the distribution $\pi$ without
changing the remaining coordinates, $\omega_{[n] \setminus i}$.

Recall that the \underline{down-up walk} $\Pgd$ on $(X, \pi)$ -- also known as
the \underline{Glauber dynamics} -- updates a randomly sampled
coordinate at each step. Thus, we have
\[ \Pgd = \frac{1}{n} \sum_{i=1}^n \Quu_i.\]
There is a long sequence of works giving sharp estimates for eigenvalues of the
down-up walk \cite{KaufmanM17, KaufmanO18, DinurK17, DiksteinDFH18, AlevL20}.
We recall,
\begin{theorem}[\cite{AlevL20}, Theorem 3.1]
    Let $(X, \pi)$ be an $n$-partite simplicial complex Then,
    \[ \Gap(\Pgd) \ge \frac{1}{n}\prod_{i = 0}^n (1 - \gamma_i).\]
\end{theorem}
The following upper bound is also well-understood,
\begin{proposition}[\cite{AlevL20}, Proposition 3.3]\label{prop:lb}
    Let $(X, \pi)$ be an $n$-partite simplicial complex, where
    $|X[i]| > 1$ for all $i \in [n]$. Then, writing $\Pgd$ for the
    down-up walk on $(X, \pi)$
    \[ \Gap(\Pgd) \le \frac{2}{n}.\]
\end{proposition}
On a partite simplicial complex, the underlying reason for \cref{prop:lb} is
clear: To be able to get really close
to the stationary distribution $\pi$, one needs to perform an update to nearly
all the coordinates, which by the coupon-collector bound requires at least
$\Theta(n\log n)$ random updates, hence after $n$-steps the
distribution of the random walk can still be
far away from the stationary distirbution.

Given an ordering $\vec s = (s(1), \ldots, s(n))$ of $[n]$, we define the
\underline{sequential sweep} $\Psq^{(\vec s)}$ as the random walk operator
\[ \Psq^{(\vec s)} := \Quu_{s(1)} \cdots \Quu_{s(n)},\]
which updates the coordinates $s(1), \ldots, s(n)$ sequentially one after the
other. When $\vec s = \parens*{1, \ldots, n}$, i.e.~$\vec s$ is the canonical
order of $[n]$, we will simply write $\Psq := \Psq^{(\vec s)}$. We note that
the weighting $\pi$ of the simplicial complex is stationary for
$\Psq$, since $\pi
\Quu_i = \pi$ for all $i = 1,\ldots, n$. Thus,
\[ \pi \Psq = \pi\Quu_1 \Quu_2 \cdots \Quu_n = \pi \Quu_2 \cdots \Quu_n =
\cdots = \pi \Quu_n = \pi.\]
\begin{remark}
    Notice that we have,
    \[ \Psq^* = (\Quu_1 \cdots \Quu_n)^* = \Quu_n^* \cdots \Quu_1^* = \Quu_n
    \cdots \Quu_1,\]
    In particular the random-walk $\Psq$ is not necessarily reversible!
\end{remark}

It has already been observed in \cite{Amit91, Oppenheim21}, that the
update operators
$\Quu_i$ are orthogonal projections.
\begin{proposition}\label{prop:proj}
    For all $i = 1, \ldots, n$ the operator $\Quu_i$ is an orthogonal
    projection to the subspace,
    \[ \Uu_i = \sspan\set*{ \uuu_{\alpha} : \alpha \in X[1, \cdots, i-1, i+1,
    \cdots, n]}, \]
    where $\uuu_\alpha(\omega) = \one[\omega \supset \alpha]$.
\end{proposition}
\begin{remark}
    We have,
    \[ \Uu_i = \sspan\set*{ \eff \in \RR^{X^{(n)}} : \eff(\omega) =
            \eff(\bar\omega)~~\textrm{for all}~~\omega,\bar\omega \in
            X^{(n)}~\textrm{such that}~\omega_{[n] \setminus i} =
            \bar\omega_{[n]
    \setminus i}}.\]
\end{remark}
\begin{proof}[Proof of \cref{prop:proj}]
    We note that for each $i \in [n]$, the vectors $\uuu_\alpha$ for $\alpha
    \in X^{(i)}$ are disjointly supported. Thus, by \cref{fac:proj},
    the orthogonal projector $\Quu_{\Uu_i}$ to the subspace $\Uu_i$ is given by,
    \[ \Quu_{\Uu_i} = \sum_{\alpha \in X^{(i)}} \frac{\uuu_\alpha
    \uuu_\alpha^*}{\norm{\uuu_\alpha}^2_{\pi}}. \]
    Further, we note
    \[ \norm{\uuu_\alpha}^2_{\pi} = \sum_{\omega \in X^{(n)}}
        \pi(\omega) \cdot \one[\omega \supset \alpha] = \Pr_{\bar\omega
    \sim \pi}\sqbr*{\bar\omega \supset \alpha}. \]
    Now, fix $\omega$ and $\omega'$. We have,
    \[\Quu_{\Uu_i}(\omega, \omega') = \one_{\omega}^\top \cdot \Quu_{\Uu_i}
        \cdot \one_{\omega'} = \sum_{\alpha \in X^{(-i)}} \frac{\one[\omega
            \supset \alpha] \cdot \pi(\omega') \cdot \one[\omega' \supset
        \alpha]}{\Pr_{\bar\omega \sim \pi}\sqbr*{\bar\omega \supset \alpha}} =
        \begin{cases} 0 & \textrm{if } \omega_{[n] \setminus i} \ne
            \omega'_{[n]\setminus i},\\ \frac{\pi(\omega')}{\Pr_{\bar\omega \sim
            \pi}\sqbr*{\bar\omega \supset \omega_{[n]\setminus i}}} &
            \textrm{otherwise.}
    \end{cases}\]
    Here we have used the observations,
    \[\one_\omega^\top \uuu_\alpha = \one[\omega \supset \alpha] =
        \one[\omega_{[n] \setminus i} =
        \alpha]~~\textrm{and}~~\uuu_\alpha^* \one_{\omega'} =
        \pi(\omega') \cdot \one[\omega' \supset \alpha] = \pi(\omega') \cdot
    \one[\omega'_{[n]\setminus i} = \alpha].\]
    Noting that there can be at most one $\alpha \in X[1,\ldots, i-1, i+1,
    \ldots, n]$ such that $\alpha
    = \omega_{[n] \setminus i}= \omega'_{[n] \setminus i}$ yields the
    equality. Now, note that by Bayes' law
    \cref{eq:b-rule},
    \[ \Quu_{\Uu_i}(\omega, \omega') = \Pr_{\bar\omega \sim \pi}[\bar\omega=
        \omega' \mid \bar\omega \supset \omega_{[n] \setminus i}] =
    \Quu_i(\omega, \omega').\]
\end{proof}

\subsection{Colored Random Walks on Partite Simplicial
Complexes}\label{ss:colrw}
It will also be important to recall the \underline{colored random
walks}, introduced in \cite{DiksteinD19, GurLL22}.  Let  $\alpha \in
X$ be any given face of an
$n$-partite weighted simplicial complex. Let $I, J \subset [n] \setminus
\typ(\alpha)$ be two disjoint sets, i.e.~$I \cap J = \varnothing$. We define the
\underline{$(I, J)$ colored random walk} $\Cee_\alpha^{I \to J} \in
\RR^{X_\alpha[I]
\times X_\alpha[J]}$ from $X_\alpha[I]$ to
$X_\alpha[J]$ as follows:
\begin{equation}\Cee_\alpha^{I \to J}(\tau_I, \tau_J) = \pi^{(\alpha \oplus
    \tau_I)}_{J}(\tau_J) = \Pr_{\omega \sim
    \pi}[\omega_{J} = \tau_J \mid \omega_I = \tau_I,
    \omega_{\typ(\alpha)} = \alpha]~~\textrm{for all}~~\tau_I \in
    X_\alpha[I], \tau_J \in X_\alpha[J].\label{eq:c-def}
\end{equation}

We observe,
\begin{obsv}
    The matrix $\Cee_\alpha^{I \to J}$ is a row-stochastic matrix satisfying
    $\pi_I^{(\alpha)} \Cee_\alpha^{I \to J} = \pi_J^{(\alpha)}$.
    Similarly, $(\Cee_\alpha^{I \to J})^* = \Cee_\alpha^{J \to I}$
\end{obsv}
\begin{proof}
    The row-stochasticity of $\Cee_\alpha^{I \to J}$ is evident from the
    definition. The second claim is an immediate consequence of the Bayes' rule
    \cref{eq:b-rule} and \cref{eq:p-def},
    \[ [\pi_I^{(\alpha)}\Cee_\alpha^{I \to J}](\tau_J) = \sum_{\tau_I \in
        X_\alpha[I]} \pi^{(\alpha)}_I(\tau_I) \pi^{(\alpha \oplus
        \tau_I)}_J(\tau_J) = \sum_{\tau_I \in X_\alpha[I]}
        \pi^{(\alpha)}_{I \cup
    J}(\tau_I \oplus \tau_J) = \pi_J^{(\alpha)}(\tau_J).\]

    For the second claim we appeal to \cref{prop:adjoint-defn}. For
    all $\tau_I \in
    X_\alpha[I]$ and $\tau_J \in X_\alpha[J]$, we have
    {\small\[ (\Cee_\alpha^{I \to J})^*(\tau_J, \tau_I)
            = \pi^{(\alpha\oplus \tau_I)}_J(\tau_J) \cdot
            \frac{\pi_I^{(\alpha)}(\tau_I)}{\pi_J^{(\alpha)}(\tau_J)} =
            \frac{\pi^{(\alpha)}_{I \cup J}(\tau_I \oplus
            \tau_J)}{\pi^{(\alpha)}_J(\tau_J)} = \pi_I^{(\alpha
            \oplus \tau_J)}(\tau_I) =
    \Cee_\alpha^{J \to I}(\tau_J, \tau_I).\]}
\end{proof}
\begin{remark}\label{rem:comb}
    Combinatorially, we can think of $\Cee_\alpha^{I \to J}$ as
    describing a random
    walk on a bipartite graph $G_{I, J}^{(\alpha)}$ with left part
    $X_\alpha[I]$ and right part
    $X_\alpha[J]$, where each edge $(\tau_I, \tau_J)$ is given the
    weight $\pi_{I
    \cup J}^{(\alpha)}(\tau_I \oplus \tau_J)$. In particular, writing $\Emm$ for
    the random walk matrix of $G_{I, J}^{(\alpha)}$ we have
    \[ \Emm =
        \begin{pmatrix} 0 & \Cee_\alpha^{I \to J}\\ \Cee_\alpha^{J \to I} & 0
    \end{pmatrix}.\]
    In particular, we have $\sigma_2(\Cee_\alpha^{I \to J}) < 1$ if and only if
    $G^{(\alpha)}_{I, J}$ is connected.
\end{remark}

\begin{remark}
    In \cite{DiksteinD19}, the notation $\Emm_\alpha^{I, J}$ is used
    for what we have
    called $\Cee_\alpha^{I \to J}$, whereas \cite{GurLL22} uses the notation
    $\Aye_{I, J}$.
\end{remark}

\begin{remark}
    The colored random walk $\Cee_\alpha^{I
    \to J}$ is the partite analogue of what is generally referred as swap-walks
    $\Ess_{|I|, |J|}$. We refer the reader to \cite{AlevJT19, DiksteinD19} for
    more information on these graphs.
\end{remark}
Let $I, J \subset [n]$ be two subsets satisfying $I \cap J = \varnothing$.
Suppose for some $\alpha \in X[(I \cup J)^c]$ the marginal distributions
$\pi_I^{(\alpha)}$ and $\pi_J^{(\alpha)}$ are independent, i.e.~we
have
\[ \pi_{I \cup J}^{(\alpha)}(\tau_I \oplus \tau_J) =
    \pi_I^{(\alpha)}(\tau_I) \cdot \pi_J^{(\alpha)}(\tau_J)~~\textrm{for
all}~~\tau_I \in X_\alpha[I], \tau_J \in X_\alpha[J].\]
It is easy to see, that in this case $\Cee_\alpha^{I \to J} = \one
\pi_J^{(\alpha)}$. In particular, we have\footnote{the first equality
    is the consequence of $\Cee_\alpha^{I \to J}$
being a rank one matrix}
\[ \sigma_2(\Cee_\alpha^{I \to J}) = \sigma_2(\one\pi_J^{(\alpha)})=
    0~~\textrm{and}~~\Div(\mu \Cee_\alpha^{I
    \to J}~\|~\nu \Cee_\alpha^{I \to J})
= \Div(\pi_J^{(\alpha)}~\|~\pi_J^{(\alpha)}) = 0,\]
for all $\mu,\nu \in \triangle_{X_\alpha[I]}$. We will consider
relaxations of both statements
to consider the proximity of $\pi_I^{(\alpha)}$ and $\pi_J^{(\alpha)}$ to being
product distributions. Thus, we define the \underline{$(I, J)$
    variance contraction
parameter} $\ee^{I \to
J}$ as,
\begin{equation}
    \ee^{I \to J} := \ee^{I \to J}_{(X, \pi)} = \max_{\alpha \in X[(I \cup
    J)^c]} \sigma_2\parens*{\Cee_\alpha^{I \to J}}\label{eq:gsplit}
\end{equation}
We also define, the related \underline{$(I, J)$ entropy contraction parameter}
\begin{equation}\label{eq:dsplit}
    \eta^{I \to J} = 1 - \kappa^{I \to J}~~\textrm{where}~~\kappa^{I \to J} =
    \sup\set*{\frac{\Div(\mu \Cee_\alpha^{I \to
        J}~\|~\pi_{J}^{(\alpha)})}{\Div(\mu~\|~\pi_I^{(\alpha)})} : \alpha \in
    X[(I \cup J)^c]~\textrm{and}~\mu \in \triangle_{X_\alpha[I]}}
\end{equation}
\begin{remark}
    There is a high-level similarity between the way we have defined the
    constant $\eta^{I \to J}$ and the \ref{eq:mlsc} in \cref{ss:rws}.
    Since the operators
    $\Cee_\alpha^{I \to J}$ are rectangular, we will not use the same
    notation
    for them.
\end{remark}
\begin{remark}
    Let $\alpha \in X[(I \cup J)^c]$ be arbitrary. For any $\eff \in
    \RR^{X_\alpha[J]}$ satisfying $\Exp_{\pi_J^{(\alpha)}} \eff = 0$, we have
    \[ \norm*{ \Cee_\alpha^{I \to J} \eff}_{\pi_I^{(\alpha)}} \le \ee^{I \to J}
    \cdot \norm{\eff}_{\pi_J^{(\alpha)}}.\]
    Similarly, for any $\mu \in \triangle_{X_\alpha[I]}$ we have
    \[ \Div( \mu \Cee_\alpha^{I \to J}~\|~\pi_J^{(\alpha)}) \le (1 -
    \eta^{I \to J}) \cdot \Div(\mu~\|~\pi_I^{(\alpha)}).\]
\end{remark}

\subsubsection*{$\vec \ee$-product distributions and $\ee^{I \to J}$}
When dealing with the colored random walks $\Cee_\alpha^{I \to J}$ an
alternative notion of high-dimensional expansion introduced in
\cite{GurLL22} proves more convenient: Let $(X, \pi)$ be an
$n$-partite simplicial complex with parts
$X[1], \cdots, X[n]$. The distribution $\pi$ is called
\underline{$\ee$-product} if,
\[ \sigma_2(\Cee_\varnothing^{i \to j}) \le \ee~~\textrm{for all}~~i, j \in
[n]~\textrm{where}~i \ne j.\]

We extend this definition in the following manner: Let $\vec \ee = (\ee_0,
\ldots, \ee_{n-2})$ and $(X, \pi)$ be an arbitrary $n$-partite simplicial
complex. We say the distribution $\pi$ is \underline{$\vec \ee$-product} if for
all $\alpha \in X^{(\le n -2)}$ we have,
\[ \sigma_2(\Cee_\alpha^{i \to j}) \le \ee_{|\alpha|}~~\textrm{for all}~~i, j
\in [n]\setminus \alpha~\textrm{where}~i \ne j.\]
Equivalently, $\pi$ is $\vec\ee$-product if and only if $\pi^{(\alpha)}$ is
$\ee_{|\alpha|}$-product for all $\alpha \in X^{(\le n- 2)}$.
We recall the following result of \cite[Claim 7.10]{DiksteinD19},
\begin{proposition}\label{prop:connected}
    Let $(X, \pi)$ be an $n$-partite simplicial complex in which every link
    $G_{\alpha}$ is connected -- equivalently, $(X, \pi)$ is a one-sided
    $\vec \gamma$-local spectral expander where $\gamma_i < 1$ for all $i$.
    Then, the bipartite graph $G_{\set i, \set j}$
    described in \cref{rem:comb} is connected. In particular,
    there exists a vector $\vec\ee = (\ee_i)_{i = 0}^{n-2}$ such
    that $\ee_i < 1$ for all $i$, for which the distribution $\pi$
    is $\vec\ee$-product.
\end{proposition}
They also prove the following down-trickling result, \cite[Corollary
7.6]{DiksteinD19}
\begin{theorem}\label{thm:yod}
    Let $(X, \pi)$ be an $n$-partite simplicial complex all of whose
    links are connected. Suppose
    $\gamma_{n-2}(X, \pi) \le \frac{\ee}{(n-2)\ee + 1}$.  Then,
    for every $\alpha \in X^{(k)}$ and every $i, j \in [n]\setminus
    \typ(\alpha)$, we have $\sigma_2(\Cee_\alpha^{i \to j}) \le
    \frac{\ee}{k \ee + 1}$.
\end{theorem}
Additionally, we also make the following crude observation,
\begin{obsv}\label{obsv:impl}
    Let $(X, \pi)$ be a weighted $n$-partite simplicial complex which is a
    one-sided $\vec \gamma$-local spectral expander.  Then,
    $\pi$ is $\vec \ee$-product where $\ee_i = (n-1-i) \cdot
    \gamma_i$.

    Equivalently, if $\pi$ is $\vec \ee$-spectrally independent, then it is
    also $\vec \ee$-product.
\end{obsv}
\begin{remark}\label{rem:difficult}
    While this blow-up in dimension seems too pessimistic, we note that in a
    worst-case scenario this is unavoidable (qv.~\cref{xmpl:ce}).
    However, the power of \cref{thm:yod} is that it allows us to
    circumvent this difficulty by assuming very strong expansion for
    the upper levels. It is easy to see that when
    $\gamma_{n-2}(X,\pi) \le \ee_0/2n$, \cref{thm:yod} implies that
    $\sigma_2(\Cee_\alpha^{i \to j}) \le \ee_0/n$ for all $\alpha$.
    This is because we can instantiate \cref{thm:yod} with $\ee
    :=\ee_0/n$ in this case.
\end{remark}
\begin{proof}[Proof of \cref{obsv:impl}]
    For simplicity, we will assume $\alpha = \varnothing$. Let $a, b \in [n]$
    be two distinct elements and $\eff: X[a] \to \RR$ and
    $\gee:X[b] \to \RR$ be functions satisfying the following
    conditions:
    \begin{enumerate}
        \item $\sigma_2(\Cee_\varnothing^{a \to b}) = \inpr*{\eff,
                \Cee_\varnothing^{a \to b}
            \gee}_{{\pi}_{\set a}}$,
        \item  $\norm{\eff}_{\pi_{\set a}}^2 =
            \norm{\gee}_{\pi_{\set b}} = 1$,
        \item $\inpr*{\eff, \one}_{\pi_{\set a}} = \inpr*{\gee,
            \one}_{\pi_{\set b}} = 0$.
    \end{enumerate}
    We note that the existence of $\eff$ and $\gee$ is guaranteed by
    \cref{fac:rsig}.

    Then, note
    \begin{align*}
        \sigma_2(\Cee_\varnothing^{a \to b})
        ~=~\inpr*{\eff, \Cee_\varnothing^{a \to b} \gee}_{\pi_{\set
        a}}
        &~=~\sum_{x \in X[a]} \pi_{\set a}(x) \eff(x) \cdot
        [\Cee_\varnothing^{a \to b} \gee](x),\\
        &~=~\sum_{x \in X[a]} \pi_{\set a}(x) \eff(x) \cdot
        \sum_{y \in X[b]} \pi^{(x)}_{\set b}(y) \gee(y),\\
        &~=~\sum_{x \oplus y \in X[a, b]}\pi_{\set{a, b}}(x
        \oplus y) \cdot \eff(x)\gee(y),
    \end{align*}
    where the last inequality follows from \cref{eq:b-rule}. Then,
    notice that $\pi_{\set{a, b}}(x \oplus y) = \pi_2(\set{a,
    b}) \cdot \binom{n}{2}$. We write, $\Ff: X^{(1)}_\alpha \to
    \RR$ and $\Gg: X^{(1)}_\alpha \to \RR$  for the extensions of $\eff$ and
    $\gee$ to $X^{(1)}$, i.e.~
    \[\Ff(x) = \one[\typ(x) = a] \cdot \eff(x)~~\textrm{and}~~\Gg(x) =
    \one[\typ(x) = b] \cdot \gee(x)~~\textrm{for all}~~x \in X^{(1)}\]
    Notice that $\norm{\Ff + \Gg}^2_{\pi_1} = \frac{2}{n}$ since for each $x
    \in X[i]$ we have $\pi_{\set i}(x) =  n \cdot \pi_1(x)$.
    Thus, we have
    \[ \sigma_2(\Cee_\varnothing^{a \to b}) = \binom{n}{2} \cdot
        \inpr*{\Ff + \Gg,
            \Emm_\varnothing
        (\Ff+ \Gg)}_{\pi_1} \le \frac{\lambda_2(\Inf_\varnothing)}{n-1} \cdot
    \frac{2}{n} \cdot \binom{n}{2} \le \ee_0,\]
    notice that here we have used that $\Ff + \Gg \in \Tt_\varnothing^\perp$
    since $\Exp_{\pi_{\set a}}\eff = \Exp_{\pi_{\set b}} \gee = 0$.
\end{proof}

We now recall the following result by \cite{GurLL22, DiksteinD19} which allows
us to bound $\ee^{I \to J}$ assuming $\vec \ee$-product property,
\begin{theorem}
    Let $(X, \pi)$ be an $n$-partite simplicial complex where $\pi$ is $\vec
    \ee$-product with $\ee_i \le \ee$ for all $i = 0, \ldots, n-2$. We have,
    \[ \sigma_2(\Cee_\alpha^{I \to J})^2 \le |I||J| \cdot \ee^2.\]
\end{theorem}
We will prove an improved version of this result in \cref{ss:sg},
(qv.~\cref{thm:cwadv}) which will show that whenever all the links
$G^{(\alpha)}$ in $X^{(\le n - 2)}$ connected, we have $\sigma_2(\Cee_\alpha^{I
\to J}) < 1$ for all $\alpha \in X^{(\le n -2)}$ and disjoint sets $I, J
\subseteq [n] \setminus \typ(\alpha)$ where $I \cup J \subset [n] \setminus
\typ(\alpha)$.

\subsection*{$\eta^{I \to J}$ and entropic independence}
We recall the concept of \underline{entropic independence}, \cite{AnariJKP22}.
Let $(X, \pi)$ be a simplicial complex of rank $n$ and write $\Dee_{k \to
\ell} \in \RR^{X^{(n)} \times X^{(\ell)}}$ for the operator defined
in the following manner,
\[ \Dee_{n \to \ell}(\omega, \alpha) = \frac{\one[\alpha \subset
    \omega]}{\binom{n}\ell}~~\textrm{for
all}~~\alpha \in X^{(\ell)}, \omega \in X^{(n)}.\]
We observe $\pi_n \Dee_{n \to \ell} = \pi_\ell$. A distribution $\mu
\in \triangle_{X^{(n)}}$ is said to be
\underline{$a^{-1}$-entropically
independent} for $a \in (0,1]$, if for all $\nu \in X^{(n)}$ one has
\[ \Div(\nu \Dee_{n \to 1}~\|~\mu\Dee_{n \to 1}) \le \frac{1}{a \cdot
n}\Div(\nu~\|~\mu).\]
Entropic independence of the distribution $\pi$ can be seen as an entorpic
generalization for spectral independence or local-spectral expansion
(qv.~\cref{ss:links}). Indeed
the following is proven in \cite{AnariJKP22},
\begin{theorem}[Simplified Version of Theorem 5 in \cite{AnariJKP22}]
    Let $(X, \pi)$ be a simplicial complex of rank $n$. Suppose for every
    $\alpha \in X^{(n- 2)}$, the distribution $\pi^{(\alpha)}$ is
    $a^{-1}$-entropically independent, where $a^{-1} \in \NN$. Then,
    writing $\ell := k - \lceil a^{-1}
    \rceil$, we have
    \[ \Div(\nu \Pgdl\ell~\|~\pi) \le (1 - \kappa) \cdot
    \Div(\nu~\|~\pi)~~\textrm{for all}~~\nu \in \triangle_{X^{(n)}},\]
    where $\Pgdl{\ell} = \Dee_{n \to \ell} \Dee_{n \to \ell}^*$ and $\kappa=
    \binom{n - \ell
    }{a^{-1}}/\binom{n}{a^{-1}}$. In particular, $\EC(\Pgdl\ell) \ge \kappa$.
\end{theorem}
We note that $\Pgd = \Pgdl{n-1}$. There is a similarity in the way we have
defined our constants $\eta^{J \to I}$ and the concept of $a^{-1}$-entropic
independence: For our result that bounds the \ref{eq:mlsc} $\EC(\Psq)$
(qv.~\cref{thm:ecc}), we will need to find the best constants $\eta^{[j - 1] \to
j}$ that satisfy
\[ \Div(\mu\Cee_\alpha^{[j- 1] \to j}~\|~\pi_{[j-1]}^{(
    \alpha)}\Cee_\alpha^{[j-1] \to j}) = \Div(\mu\Cee_\alpha^{[j-1]
        \to j}~\|~\pi_{\set
    j}^{(\alpha)}) \le
    (1 - \eta^{[j-1] \to j})\Div(\mu~\|~\pi_{[j-1]}^{(\alpha)})
\]
for all $\alpha \in X[j+1, \ldots, n]$ and for all $\mu \in
\triangle_{X_\alpha[1, \ldots, j-1]}$.

\subsection{Ramanujan Complexes}\label{sub:ramanujan_complexes} 

We give here a brief description of Ramanujan complexes, which were
defined in \cite{LubotzkySV05,Li2004ramanujan} and constructed in \cite{Li2004ramanujan,LubotzkySV05b,First16,Evra2018RamanujancomplexesGolden,dalal2026ramanujan}. Let $p$
be a prime and $d\geq3$.
The Bruhat-Tits building $\mathcal{B}=\mathcal{B}_{d,p}$ is a
$(d-1)$-dimensional
simplicial clique complex, whose underlying graph $\mathcal{B}^{(1)}$
is defined as follows: its vertices correspond to classes of
$\mathbb{Z}_{p}$-lattices
in $\mathbb{Q}_{p}^{d}$ up to scaling by $\mathbb{Q}_{p}^{\times}$,
and two classes $C,C'\in\mathcal{B}^{0}$ are neighbors in $\mathcal{B}^{(1)}$
if $pL<L'<L$ for some representatives $L\in C$ and $L'\in C'$.
By the transitive action of $G=\mathsf{PGL}_{d}(\mathbb{Q}_{p})$ on the classes
of lattices we obtain $\mathcal{B}^{0}\cong G/K$, where
$K=\mathsf{PGL}_{d}(\mathbb{Z}_{p})$
is the stabilizer of $\mathbb{Z}_{p}^{d}$. 

We give the vertices colors
in $\nicefrac{\mathbb{Z}}{d\mathbb{Z}}$ by
$\mathrm{col}(gK)=\mathrm{val}_{p}\det g$,
and the (directed) edges colors in
$(\nicefrac{\mathbb{Z}}{d\mathbb{Z}})\backslash0$
by $\mathrm{col}(v\rightarrow w)=\mathrm{col}(w)-\mathrm{col}(v)$.
The colored adjacency operators
\[
    \left(\Aye_{i}f\right)\left(v\right)=\sum\nolimits_{w:\mathrm{col}(v\rightarrow
    w)=i}f(w)\qquad\left(1\leq i\leq d-1\right)
\]
are called the \emph{Hecke operators }on $\mathcal{B}$. It is easy
to see that $\Aye_{i}^{*}=\Aye_{d-i}$, and a nontrivial fact is that all
$\Aye_{i}$ commute with each other, which in particular implies that
they are normal.

The clique complex $\mathcal{B}$ obtained from the graph $\mathcal{B}^{(1)}$
is $d$-partite, every facet having one vertex of each color. By a
result of Tits, the color-preserving automorphisms of $\mathcal{B}$
are the elements of $G'=\mathsf{PSL}_{d}(\mathbb{Q}_{p})$, and if
$\Gamma\leq G'$
is a torsion-free lattice, then $X=\Gamma\backslash\mathcal{B}$ is
a finite $d$-partite quotient of $\mathcal{B}$. 
Recall that a $k$-regular graph is called a Ramanujan graph if every adjacency eigenvalue $\lambda$ either satisfies $|\lambda|=k$, or belongs to the $L^2$-spectrum of the $k$-regular tree, which is $[-2\sqrt{k-1},2\sqrt{k-1}]$. Similarly, the complex $X$
is called a \emph{Ramanujan complex }if every
$\lambda\in\mathrm{Spec}(\Aye_{i}|_{L^{2}(X^{0})})$
satisfies either
\begin{align*}
    \left|\lambda\right| & =\deg \Aye_{i}=\left[
        \begin{smallmatrix}d\\
            i
    \end{smallmatrix}\right]_{p},~\text{or}\\
    \lambda &
    \in\mathrm{Spec}\parens*{ \Aye_{i}|_{L^{2}(\mathcal{B}^{0})} }=\left\{
        p^{\frac{i(d-i)}{2}}(z_{1}+\ldots+z_{d})\,\middle|\,\left|z_{i}\right|=1\text{
    and }z_{1}\cdot\ldots\cdot z_{d}=1\right\} .
\end{align*}
A useful property of $\mathcal{B}$ is that the link of every vertex
is isomorphic to the spherical building of $\mathsf{PGL}_{d}(\mathbb{F}_{p})$, which we denote by $\mathbb{P}^{d-1}(\mathbb{F}_{p})$. This is a pure, $(d-1)$-partite
    simplicial complex of flags in $\mathbb{F}_{p}^{d}$. Namely, the
    vertices of $\mathbb{P}^{d-1}(\mathbb{F}_{p})$ are the
    nontrivial subspaces
    $0\lneq V\lneq\mathbb{F}_{p}^{d}$, with types given by
    $\typ(V) = \dim_{\mathbb{F}_{p}}(V)\in\left\{
    1,\ldots,d-1\right\}$. Similarly, the facets $(V_1,\ldots, V_{d-1})$
    correspond to chains
    $V_{1}< V_{2}<\cdots< V_{d-1}$.

%
%

\section{A Spectral Gap Bound for $\Psq$}\label{ss:sg}
Our main theorem in this section will be,
\begin{theorem}\label{thm:csv}
    Let $(X, \pi)$ be an $n$-partite simplicial complex, $\vec s =
    (s(1), \ldots,
    s(n))$ be any ordering of $[n]$, and the parameters
    $\ee^{I \to J}$ be defined as in \cref{eq:gsplit} for all $I, J \subset
    [n]$ with $I \cap J = \varnothing$. Then,
    \[ \sigma(\Psq^{(\vec s)})^2 \le 1-  \prod_{j = 2}^n \parens*{ 1 -
    (\ee^{s([j-1]) \to s(j)})^2},\]
    where we have written $S(A) = \set{s(a) : a \in S}$. Equivalently,
    \[ \Gap(\Psq^{(\vec s)}) \ge 1 - \sqrt{ 1 - \prod_{j = 2}^n \parens*{1 -
    (\ee^{s([j-1]) \to s(j)})^2}}.\]

    In particular, if $\pi$ is $\vec\ee$-product, we have
    \[ \sigma(\Psq^{(\vec s)})^2 \le 1 - \prod_{j = 2}^n \prod_{p = 0}^{j -2} (1
    - \ee_p^2).\]
    Equivalently,
    \[ \Gap(\Psq^{(\vec s)}) \ge 1 - \sqrt{1 - \prod_{j = 2}^n \prod_{p =
    0}^{j-2} (1 - \ee_p)^2}\]
\end{theorem}
In conjunction with \cref{obsv:impl} and \cref{thm:yod},
\cref{thm:csv} immediately implies the
following
\begin{corollary}\label{cor:coraa}
    Let $(X, \pi)$ be a link-connected $n$-partite simplicial complex
    that is a $\vec \gamma$-local spectral expander with
    $\gamma_i \le \ee/(n -1 -i )$ (i.e.~$\pi$ $\vec\ee$-product, with $\ee_i
    \le \ee$), then
    \[ \sigma_2(\Psq) \le \frac{n \cdot \ee}{\sqrt
        2}~~\textrm{and}~~\Gap(\Psq) \ge 1 -
    \frac{n \cdot \ee}{\sqrt 2}.\]
    If we have $\gamma_{n-2} \le \frac{\ee}{2n}$ and all links of
    $(X,\pi)$ are connected, then
    \[ \sigma_2(\Psq) \le \frac{\ee}{\sqrt 2}~~\textrm{and}~~\Gap(\Psq) \ge 1 -
    \frac{\ee}{\sqrt 2}.\]
\end{corollary}
For the second statement we note that \cref{thm:yod} implies that we
are completely $\ee/n$ product assuming $\gamma_{n-2} \le \ee/2n$ --
see \cref{rem:difficult}. In particular, the $\ee$ in the statement
of \cref{thm:yod} and the $\ee$ above differ by a factor of $n$. The
above corollary shows that for constant $n$, having constant local
spectral expansion guarentees that $\Gap(\Psq)$ is close to 1 -- thus it allows
us to conclude that $\Psq$ mixes rapidly in the complexes of
\cite{LubotzkySV05, KaufmanO18b}.

As already advertised in the introductory part of this paper, in
conjunction with \cref{thm:spec-mix-bd}, shows that for very strong
$\vec \ee$-product complexes or $\vec \gamma$-local spectral
expanders \cref{cor:coraa} immediately implies that within $O(1)$
sweeps $\Psq$ will converge to the stationary distribution:
\begin{corollary}\label{cor:corbb}
    Let $(X, \pi)$ be a link-connected $n$-partite simplicial complex
    with connected links that is a $\vec \gamma$-local spectral
    expander. If we have $\gamma_{n-2} \le \frac{\ee}{2n}$ and all
    links of $(X,\pi)$ are connected, then for any $\delta > 0$, we have
    \[ T(\delta, \Pii) \le \left\lceil \frac{ \log\parens*{
                \parens*{\delta \cdot \sqrt{\min_{\omega \in X^{(n)}}
    \pi(\omega)}}^{-1}}}{\log( \sqrt 2 / \ee ) }\right\rceil.\]
    In particular, for any fixed $\delta > 0$ as $\ee$ tends to $0$,
    $T(\delta, \Pii)$ tends to 1.
\end{corollary}

We will prove \cref{thm:csv} using \cref{thm:prprod}.  To be able apply
\cref{thm:prprod}, we will need some preparatory work:
\begin{enumerate}
    \item For all $T \subset [n]$, we will characterize the
        intersections $\bigcap_{t \in T} \Uu_t$, where
        $\Uu_t$ are as defined in \cref{prop:proj}, i.e.~$\Uu_t = \im(\Quu_t)$.
        This will be done in \cref{prop:ints}.
    \item We will get a bound on $\sin^2(\Uu_p, \bigcap_{t \in T}
        \Uu_t)$ for all $p \in
        [n]$ and $T \subset [n] \setminus p$ in terms of the parameters $\ee^{I
        \to J}$, defined in \cref{eq:gsplit}. This will be done in
        \cref{thm:angl}.
    \item We will get a bound on the parameters $\ee^{I \to J}$ under
        the assumption
        that $\pi$ is $\vec \ee$-product. This will be done in
        \cref{thm:cwadv}.
\end{enumerate}
We state these results now.
\begin{proposition}\label{prop:ints}
    For $(X, \pi)$ a weighted and link-connected $n$-partite simplicial
    complex. Let $\Uu_1,
    \ldots, \Uu_n$ be as defined in \cref{prop:proj}. For all $T \subset [n]$,
    we define
    \[ \Uu_{T} := \sspan\set*{ \uuu_\alpha : \alpha \in X\sqbr{T^c}}. \]
    Then, $\Uu_T = \bigcap_{t \in T} \Uu_t.$
\end{proposition}
\begin{remark}\label{rem:equivalent-formulation}
    We have,
    \[ \Uu_I = \sspan\set*{ \eff \in \RR^{X^{(n)}} : \eff(\omega) =
            \eff(\bar\omega)~~\textrm{for all}~~\omega,\bar\omega \in
            X^{(n)}~\textrm{such that}~\omega_{[n] \setminus I} =
            \bar\omega_{[n]
    \setminus I}}.\]
    In particular, $\Uu_I$ is the space of vectors $\eff \in
    \RR^{X^{(n)}}$ whose entries are completely determined by the coordinates
    in $[n] \setminus I$.
\end{remark}

We note that a special case of this result was already proven in
\cite[Proposition 3.6]{Oppenheim21}. We will present a proof for
\cref{prop:ints}, in \cref{ss:ints}.

Then, we show the geometric interpretation for the quantities $\ee^{I \to J}$ as
follows,
\begin{theorem}\label{thm:angl}
    Let $(X, \pi)$ be an $n$-partite simplicial complex. Let $I, J \subset [n]$
    be any pair of disjoint sets, i.e.~$I \cap J = \varnothing$.

    Then, $\cos(\Uu_I, \Uu_J) \le \ee^{I \to J}$,
    where the subspaces $\Uu_I, \Uu_J$ are as defined in \cref{prop:ints} and
    the parameters $\ee^{I \to J}$ are as defined in \cref{eq:gsplit}.
\end{theorem}
We will prove \cref{thm:angl} in \cref{ss:angl}.

Finally, we state our bound on the parameters $\ee^{I \to J}$,
\begin{theorem}\label{thm:cwadv}
    Let $(X, \pi)$ be an $n$-partite simplicial complex where $\pi$ is $\vec
    \ee$-product. We have,
    \begin{equation}\sigma_2(\Cee_\alpha^{I \to J})^2 \le
        1 - \prod_{p = 0}^{|I| - 1}\prod_{q = 0}^{|J| - 1}(1-
        \ee_{|\alpha|+ p+ q}^2)\label{eq:cw}
    \end{equation}
    In particular, we have
    \[ (\ee^{I \to J})^2 \le 1 - \prod_{p = 0}^{|I| - 1}\prod_{j = 0}^{|J| -
    1}(1 - \ee_{n - |I| - |J| + p + q}^2)\]
\end{theorem}
\begin{remark}\label{rem:veryconnected}
    Notice that the RHS of \cref{eq:cw} is always strictly smaller than 1
    provided that $\pi$ is $\vec \ee$-product with $\ee_i < 1$ for all $i =
    0,\ldots, n-2$. By
    \cref{prop:connected}, this is always the case assuming all the
    links in $(X,
    \pi)$ are connected.
\end{remark}
We will present the proof for \cref{thm:cwadv} in \cref{ss:cwadv}.

The proof of \cref{thm:csv} is now immediate,
\begin{proof}[Proof of \cref{thm:csv}]
    For simplicity, we will assume $\vec s = (1, \ldots, n)$. By
    \cref{thm:prprod} we have,
    \[ \norm{ \Quu_1 \cdots \Quu_n\eff - \Quu_\star\eff}^2_{\pi} \le
        \parens*{1 -
            \prod_{j = 2}^n \sin^2\parens*{\Uu_j, \bigcap_{i = 1}^{j-1}
        \Uu_i}} \cdot
    \norm{\eff - \Quu_\star \eff}^2_\pi,\]
    where $\Quu_\star$ is the projection to $\bigcap_{j = 1}^n \Uu_j$. By
    \cref{prop:ints},
    $\Quu_\star$ is the projection to constant functions. Similarly, again by
    \cref{prop:ints}, we have
    $\bigcap_{i = 1}^{j-1} \Uu_i = \Uu_{[j-1]}$. By \cref{fac:rsig}, this
    means
    \[ \sigma_2(\Quu_1 \cdots \Quu_n)^2 \le 1 - \prod_{j = 2}^n \sin^2(\Uu_j,
    \Uu_{[j-1]}).\]
    Now, the statements follow from substituting $\sin^2(\Uu_j,
    \Uu_{[j-1]}) = 1 -
    \cos^2(\Uu_j, \Uu_{[j-1]})$, \cref{prop:ints}, and \cref{thm:cwadv}.
\end{proof}
\subsection{Proof of \cref{thm:cwadv}}\label{ss:cwadv}
\begin{proof}[Proof of \cref{thm:cwadv}]
    If $|I| = |J| = 1$, then the RHS of \cref{eq:cw} equals $\ee_{|\alpha|}$
    and thus the statement follows since $\pi$ is $\vec \ee$-product.
    Suppose without loss of generality, that $\alpha = \varnothing$ and
    $|I| \ge 2$ where we have some fixed element $i \in I$. Write,
    \begin{equation}\sigma := \max_{x \in X[i]}
        \sigma_2\parens*{\Cee_x^{(I \setminus i) \to
        J}}.\label{eq:shh}
    \end{equation}

    Let $\Ff \in X[J] \to \RR$ satisfy $\norm{\Ff}_{\pi_J} = 1$,
    $\inpr*{\Ff, \one}_{\pi_J} = 0$, and
    $\norm{\Cee_\varnothing^{I \to J} \Ff}_{\pi_I} =
    \sigma_2(\Cee_\varnothing^{I \to J})$. Then,
    \begin{align*}
        \sigma_2\parens*{\Cee_\varnothing^{I \to J}}^2 =
        \norm{\Cee_\varnothing^{I \to J} \Ff}^2_{\pi_I}
        &~=~\Exp_{\alpha_I \sim
        \pi_I}\parens*{\Exp_{\alpha_J \sim \pi_J^{(\alpha_I)}}
        \Ff(\alpha_J)}^2,\\
        &~=~\Exp_{x \sim \pi_{\set i}}\Exp_{\alpha' \sim \pi_{I \setminus
        i}^{(x)}}\parens*{\Exp_{\alpha_J \sim \pi_J^{(x \oplus
        \alpha')}} \Ff(\alpha_J)}^2,\\
        &~=~\Exp_{x \sim \pi_{\set i}}
        \norm*{ \Cee_{x}^{(I \setminus i)\to J} \Ff_x}^2_{\pi^{(x)}_{I \setminus
        i}},\\
        &~\le~\Exp_{x \sim \pi_{\set i}}\sqbr*{
            \norm*{\Ff_x^{\one}}^2_{\pi^{(x)}_{J}} +
            \sigma_2(\Cee_x^{(I \setminus i) \to
        J})^2 \cdot \norm*{ \Ff_x^{\perp \one}}^2_{\pi^{(x)}_{J}}},\\
        &~=~\sigma^2 + (1 - \sigma^2) \cdot \Exp_{x \sim \pi_{\set
        i}}\norm{\Ff_x^{\one}}_{\pi^{(x)}_J}^2
    \end{align*}
    where $\Ff_x^{\one} = \one \cdot \Exp_{\alpha} \Ff_x(\alpha_{J \setminus
    i})$.
    We note now for each $i \in X[i]$,
    \[ \norm*{\Ff^{\one}_x}^2_{\pi_J} = \parens*{ \Exp_{\alpha_J \sim
    \pi_J^{(x)}} \Ff(\alpha_J) }^2 =  [\Cee_\varnothing^{i \to J} \Ff](x)^2. \]
    Thus, using that $\inpr*{\Ff, \one}_{\pi_J} = 0$,
    \[ \Exp_{x \sim \pi_{\set 1}} \norm*{ \Ff^{\one}_x}^2_{\pi_J} =
        \norm{\Cee_\varnothing^{i \to J} \Ff}^2_{\pi_{\set i}} \le
    \sigma_2\parens*{\Cee_\varnothing^{i \to J}}^2.\]
    Thus,
    \begin{align}
        \sigma_2\parens*{\Cee_\varnothing^{I \to J}}^2
        &~\le~1 - (1 - \sigma^2)(1 - \sigma_2(\Cee_\varnothing^{i \to
        J})^2),\notag\\
        &~=~1 - \parens*{1- \sigma_2\parens*{\Cee_{x^\star}^{I \setminus i \to
        J}}^2}\parens*{1 -
        \sigma_2\parens*{\Cee_\varnothing^{i \to J}}^2}.\label{eq:rec-f}
    \end{align}
    where $x^\star = \arg\max_{x \in X[i]} \sigma_2(\Cee_x^{{I \setminus i} \to
    J})$.

    Now, we first assume $J = \set*{j}$ and $I = \set*{i_1, \ldots, i_\ell}$,
    applying \cref{eq:rec-f} recursively, we obtain
    \[ \sigma_2(\Cee_\varnothing^{I \to \set*{j}}) \le 1 - (1 -
        \sigma_2(\Cee_\varnothing^{i_1 \to j})^2) \prod_{p = k}^\ell
        \parens*{1 - \sigma_2\parens*{\Cee_{{x_{i_1} \oplus \cdots \oplus
                        {x_{i_{p-
    1}}}}}^{i_p \to \set{j}}}^2},\]
    where
    \[x_{i_1} = x^\star,~~\omega(p) = x_{i_1} \oplus \cdots
        x_{i_{p-1}},~~\textrm{and}~~x_p = \arg\max_{x \in X_{\omega(p)}[i_p]}
        \sigma_2\parens*{\Cee_{x_{i_1} \oplus \cdots \oplus
            x_{i_{p-1}} \oplus x}^{I \setminus\set*{i_1, \ldots, i_p} \to
    \set{j}}},\]
    In
    particular using the $\vec \ee$-product assumption,
    \[ \sigma_2(\Cee_\varnothing^{I \to j})^2 \le 1 - \prod_{p =
    0}^{|I|-1} (1 - \ee_p^2).\]
    Now, by using $\Cee_\varnothing^{J \to i} = \parens*{\Cee_\varnothing^{i
    \to J}}^*$  and that $\sigma_2(\Bee) = \sigma_2(\Bee^*)$ we have by the
    above
    \begin{equation}
        \sigma_2(\Cee_\varnothing^{i \to J})^2 \le 1 - \prod_{p =
        0}^{|J|-1} (1 - \ee_p^2).\label{eq:ob}
    \end{equation}
    Similarly, inserting \cref{eq:ob} inside \cref{eq:rec-f}, we obtain
    \begin{equation}\sigma_2\parens*{\Cee_\varnothing^{I \to J}}^2
        \le 1 - \prod_{p = 0}^{|J|-1} (1 - \ee_p^2) \cdot \parens*{1-
            \sigma_2\parens*{\Cee_{x^\star}^{I \setminus i \to
        J}}^2} .\label{eq:recurse}
    \end{equation}
    Suppose $I = \set*{i_1, \ldots, i_\ell}$ write $x_{i_1} = x^\star$ and
    \[\omega(q) = x_{i_1} \oplus \cdots
        x_{i_{q-1}},~~\textrm{and}~~x_q = \arg\max_{x \in X_{\omega(q)}[i_q]}
        \sigma_2\parens*{\Cee_{\omega(q) \oplus x}^{I \setminus\set*{i_1,
    \ldots, i_p} \to J}},\]

    Applying \cref{eq:recurse} repeatedly, we obtain
    \begin{equation}\sigma_2(\Cee_\varnothing^{I \to J})^2
        \le 1- \parens*{\prod_{p = 0}^{|J| - 1} (1 - \ee_p^2)} \cdot
        \prod_{q = 1}^{|I| - 1} \parens*{1 -
            \sigma_2\parens*{\Cee_{x_{i_1}\oplus
        \cdots\oplus x_{i_{q}}}^{i_{q+1}\to J}}^2}.\label{eq:risotto}
    \end{equation}
    By \cref{eq:ob}, we have
    \[ \sigma_2\parens*{\Cee_{x_{i_1} \oplus \cdots \oplus
        x_{i_q}}^{i_{q+1}\to J}}^2
        \le
    1 - \prod_{p = 0}^{|J| - 1}(1 - \ee_{q + p}^2)\]
    And inserting this in \cref{eq:risotto}, we obtain
    \begin{align*}
        \sigma_2\parens*{ \Cee_\varnothing^{I \to J} }^2
        &~\le~ 1 - \parens*{\prod_{p = 0}^{|J| - 1} (1 - \ee_p^2)} \cdot
        \parens*{\prod_{p = 0}^{|J| - 1}(1 - \ee_{p + 1}^2)}
        \cdots \parens*{\prod_{p = 0}^{|J| - 1}(1 - \ee_{p + |I| - 1}^2)},\\
        &~=~1 - \prod_{q = 0}^{|I| - 1}\prod_{p = 0}^{|I| - 1} (1 - \ee_{p +
        q}^2).
    \end{align*}
\end{proof}

\subsection{Proof of \cref{prop:ints}}\label{ss:ints}

\begin{proof}[Proof of \cref{prop:ints}]
    We proceed by induction on the number of elements in $T$. For $|T| = 1$,
    the statement is clear. Thus, we assume that
    \[ \Uu_T = \bigcap_{t \in T} \Uu_t~~\textrm{for all}~~T~~\textrm{with}~~|T|
    = \ell.\]
    Let $t \in [n] \setminus T$, and consider $T \cup \set{ t}$. First,
    fix any $\alpha \in X[ [n] \setminus (T \cup t)]$. Then, note
    \[ \uuu_\alpha= \sum_{\beta\in X[[n] \setminus T],\atop \beta \supset\alpha}
        \uuu_\beta = \sum_{\beta' \in X[1,\ldots, t-1, t+1, \ldots, n],\atop
    \beta' \supset \alpha} \uuu_{\beta'}.\]
    Thus, indeed: $\sspan\set*{\uuu_\alpha : \alpha \in X[ [n] \setminus (T
    \cup t)]}  = \Uu_{T
    \cup t} \subset \Uu_T \cap
    \Uu_t.$

    To show the reverse inclusion, we recall that a function $\eff \in
    \RR^{X^{(n)}}$ is in $\Uu_A$ if
    and only if we have
    {\small
        \[ \eff(\omega_A \oplus \omega_{[n] \setminus A}) =
            \eff(\tilde\omega_A \oplus
            \omega_{[n] \setminus A})~~\textrm{for all}~~\omega_A,
            \tilde\omega_A
            \in X[A]~\textrm{and}~\omega_{[n] \setminus A} \in X[[n] \setminus
            A]~~\textrm{such that}~~\omega_A \oplus \omega_{[n] \setminus A},
        \widetilde\omega_A \oplus \omega_{[n]\setminus A} \in X^{(n)}.\]
    }
    i.e.~the value of $\eff$ is determined only by the coordinates in
    $[n] \setminus A$ (qv.~\cref{rem:equivalent-formulation}).

    Suppose $\eff \in \Uu_T \cap \Uu_t$. We want to prove,
    \begin{equation}\eff(\omega_t \oplus \omega_T \oplus \omega_{[n]
        \setminus (T \cup t)})
        = \eff(\tilde\omega_t \oplus \tilde\omega_T \oplus \omega_{[n]
        \setminus (T \cup t)}),\label{eq:rivaldo}
    \end{equation}
    for all $\omega_t, \tilde\omega_t \in X[t]$, $\omega_T, \tilde\omega_T \in
    X[T]$ and $\omega_{[n] \setminus (T \cup t)} \in X[(T \cup t)^c]$ such that
    $\omega_t \oplus \omega_T \oplus \omega_{(T \cup t)^c}, \widetilde \omega_t
    \oplus \tilde \omega_T \oplus \omega_{(T \cup t)^c} \in X^{(n)}$.

    Since $\eff \in \Uu_t$, we have,
    \[ \eff(\omega_t \oplus \omega_T \oplus \omega_{[n] \setminus (T \cup t)})
        =
        \eff(\tilde\omega_t \oplus \omega_T \oplus \omega_{[n] \setminus (T
    \cup t)})\]
    for all $\omega_t, \tilde\omega_t \in X[t]$, $\omega_T \in X[T]$, and
    $\omega_{[n] \setminus (T \cup t)} \in X[[n] \setminus (T \cup t)]$ such
    that $\omega_t \oplus \omega_T \oplus \omega_{[n] \setminus (T \cup t)},
    \tilde\omega_t \oplus \omega_T \oplus \omega_{[n] \setminus (T \cup t )} \in
    X^{(n)}$. Similarly, since $\eff \in \Uu_T$,
    we have,
    \[ \eff(\omega_t \oplus \omega_T \oplus \omega_{[n] \setminus (T \cup t)})
        =
        \eff(\omega_t \oplus \tilde\omega_T \oplus \omega_{[n] \setminus (T
    \cup t)})\]
    for all $\omega_t \in X[t]$, $\omega_T, \tilde\omega_T \in X[T]$, and
    $\omega_{[n] \setminus (T \cup t)} \in X[[n] \setminus (T \cup t)]$ such
    that $\omega_t \oplus \omega_T \oplus \omega_{[n] \setminus (T \cup t)},
    \omega_t \oplus \tilde\omega_T \oplus \omega_{[n] \setminus (T \cup t )} \in
    X^{(n)}$. Fixing $\omega_{[n] \setminus (T \cup t)}$, this means that
    \cref{eq:rivaldo} is satisfied by any
    pair
    $(\omega_t \oplus \omega_T, \widetilde \omega_t \oplus \widetilde \omega_T)$
    which is connected by an edge in the line graph\footnote{The line graph
        $L(G)$ of a simple graph $G$ is the graph $L(G) = (E, E')$
        where $ef \in E'$
        if $|e \cap f| = 1$, i.e.~there is an edge $ef \in E'$ for any
    pair of distinct edges sharing an end point.}  $L(G_{t, T}^{(\omega_{[n]
    \setminus (T \cup t)})})$, where $G_{t, T}^{(\omega_{[n] \setminus (T
    \cup t)})}$ is defined as in \cref{rem:comb}. By \cref{rem:veryconnected},
    since the links of $(X, \pi)$ are connected, so is the graph $G_{t,
    T}^{(\omega_{[n] \setminus (T \cup t)})}$ for any $\omega_{[n]
    \setminus (T \cup t)} \in X[ [n] \setminus (T \cup t)]$. In particular, the
    graph $L(G_{t, T}^{(\omega_{[n] \setminus (T \cup t)})})$ is connected,
    meaning any pair of $\omega_t \oplus \omega_T, \tilde\omega_t \oplus
    \tilde\omega_T \in X_{\omega_{[n] \setminus (T \cup t)}}[T \cup t]$ are
    connected by a path and therefore satisfy \cref{eq:rivaldo}. This allows us
    to conclude $\eff \in \Uu_{T \cup t}$.
\end{proof}

\subsection{Proof of \cref{thm:angl}}\label{ss:angl}
\begin{proof}
    By \cref{prop:ints} we have $\Uu_I \cap \Uu_{J} = \Uu_{I \cup J}$
    and by \cref{def:cos} we have,
    \begin{equation}\label{eq:optf}
        \cos(\Uu_I, \Uu_J) = \max\set*{ \inpr*{ \etch, \gee }_\pi :
            \norm{\etch}_\pi = \norm{\gee}_\pi = 1~~\textrm{and}~~\gee \in \Uu_I
            \cap \Uu_{I \cup J}^\perp, \etch \in \Uu_{J} \cap \Uu_{I \cup
        J}^\perp }.
    \end{equation}
    Now, let $\etch$ and $\gee$ be the optimizers of the formula in
    \cref{eq:optf}. For
    convenience we introduce the functions $\Gg: X\sqbr{[n] \setminus I} \to
    \RR$ and $\Hh: X\sqbr{[n] \setminus J} \to \RR$ where,
    \begin{align*}
        \etch &~:=~\sum_{\alpha \in X\sqbr{[n] \setminus J}
        }\Hh(\alpha) \cdot \uuu_\alpha,\\
        \gee &~:=~\sum_{\bar\alpha \in X\sqbr{[n] \setminus I}}
        \Gg(\bar \alpha) \cdot \uuu_{\bar\alpha}.
    \end{align*}
    where the vectors $\uuu_\bullet$ are as defined in \cref{prop:proj},
    i.e.~$\uuu_\tau(\omega) = \one[\omega \supset \tau]$ for all $\tau \in X$.
    We observe that both $\Hh$ and $\Gg$ have the same norm as $\etch$ and
    $\gee$ respectively.
    \begin{claim}\label{cl:norm} We have,
        \[\norm{\Hh}_{\pi_{[n] \setminus J}} = \norm{\Gg}_{\pi_{[n]
        \setminus I}} = 1.\]
    \end{claim}
    \begin{proof}
        Since the proofs are identical, we only prove the statement for $\Hh$.
        By the orthogonality of the collection $\uuu_\bullet$ we have,
        \[ \norm*{\etch}_\pi^2 ~=~\sum_{\alpha \in X\sqbr{[n]
            \setminus J}} \Hh^2(\alpha)
        \cdot \norm{\uuu_\alpha}^2_{\pi}.\]
        Notice, for any $\alpha \in X[[n] \setminus J]$,
        \[ \norm{\uuu_\alpha}^2_\pi  = \sum_{\beta \in X^{(n)},\atop
                \beta \supset
        \alpha} \pi(\beta) = \pi_{[n]\setminus J}(\alpha),\]
        where we have used the definition of $\pi_{[n] \setminus J}$,
        i.e.~\cref{eq:m-def} for this equality.

        In particular,
        \[ \norm*{\etch}^2_{\pi} = \sum_{\alpha \in X\sqbr{[n] \setminus J}}
            \Hh(\alpha)^2 \cdot \pi_{[n] \setminus J}(\alpha) =
        \norm*{\Hh}_{\pi_{[n]\setminus J}}.\]
        Since $\etch$ is picked as an optimizer to \cref{eq:optf}, we have
        $\norm{\Hh}_{\pi_{[n] \setminus J}} = 1$.
    \end{proof}
    Then, note
    \begin{align*}
        \inpr*{\etch, \gee}_\pi
        &~=~\sum_{\alpha \in X\sqbr{[n] \setminus J}}\sum_{\bar\alpha \in
        X\sqbr{[n] \setminus I}} \Hh(\alpha) \Gg(\bar\alpha) \cdot
        \inpr*{ \uuu_\alpha, \uuu_{\bar\alpha}}_\pi,\\
        &~=~
        \sum_{\alpha \in X\sqbr{[n] \setminus J}}\sum_{\bar\alpha \in
        X\sqbr{[n] \setminus I}}\Hh(\alpha) \Gg(\bar\alpha) \cdot
        \pi_n(\alpha \wedge \bar\alpha).
    \end{align*}
    where we define $\alpha \wedge \bar\alpha$ as $\bot$\footnote{in
    particular, $\pi(\alpha \wedge \bar\alpha) = 0$} if there exists some $\ell
    \in [n] \setminus (I \cup J)$  such that $\alpha_\ell \neq
    (\bar\alpha)_\ell$ and otherwise we have,
    \[ (\alpha \wedge \bar\alpha)_i =
        \begin{cases} \alpha_i & \textrm{if } i
            \in I,\\
            \bar\alpha_i & \textrm{if } i \in J,\\
            \alpha_i = (\bar\alpha)_i & \textrm{if } i \in
            [n]\setminus (I \cup J)
    \end{cases}\]
    In particular, since the $(\alpha, \bar\alpha)$-term is only non-$\bot$ when
    $\alpha$ and $\bar\alpha$ agree on the coordinates $[n]\backslash (I \cup
    J)$, we can localize this sum at the intersections $\tilde \alpha \in X[
    [n] \setminus (I \cup J)]$. In particular, writing $\Hh_{\tilde
    \alpha}(\alpha_I) = \Hh(\tilde \alpha \oplus \alpha_I)$ and $\Gg_{\tilde
    \alpha}(\alpha_J) = \Gg(\tilde \alpha \oplus \alpha_J)$ for $\alpha_I \in
    X_{\tilde \alpha}[I]$ and $\alpha_J \in X_{\tilde \alpha}[J]$, we have
    \begin{align}
        \inpr*{\etch, \gee}_\pi &~=~
        \sum_{\alpha \in X[[n] \setminus J]}\sum_{\bar\alpha \in
        X\sqbr{\bar\alpha \in X[[n] \setminus I]}}
        \Hh(\alpha) \Gg(\bar\alpha) \cdot \pi_n(\alpha \wedge
        \bar\alpha),\notag\\
        &~=~\sum_{\tilde \alpha \in X\sqbr{[n]
        \setminus (I \cup J)}}\parens*{
            \sum_{\alpha_I \in X_{\tilde
            \alpha}\sqbr{I}} \sum_{\alpha_J \in
            X_{\tilde \alpha}[J]}
            \Hh_{\tilde \alpha}(\alpha_I) \Gg_{\tilde
            \alpha}(\alpha_J) \pi(\alpha_I \oplus \alpha_J \oplus \tilde
        \alpha)},\notag\\
        &~=^\vardiamond~\sum_{\tilde \alpha \in
        X\sqbr{[n] \setminus (I \cup J)}}
        \pi_{[n]\setminus (I \cup J)}(\tilde \alpha)
        \cdot \parens*{ \sum_{\alpha_J \in
            X_{\tilde \alpha}\sqbr{J}} \sum_{\alpha_I \in X_{\tilde
            \alpha}\sqbr{I}}
            \Hh_{\tilde \alpha}(\alpha_I) \Gg_{\tilde
            \alpha}(\alpha_J) \cdot \pi^{(\tilde
        \alpha)}_{I \cup J}(\alpha_I \oplus \alpha_J).},\notag\\
        &~=~\Exp_{\tilde \alpha \sim \pi_{[n]\setminus
        (I \cup J)} }
        \inpr*{\Hh_{\tilde \alpha}, \Cee_{\tilde \alpha}^{J \to I}
        \Gg_{\tilde \alpha}}_{\pi^{\tilde
        \alpha}_{I}}\label{eq:chainlaststep}
    \end{align}
    where we have used \cref{eq:b-rule} to get the equality marked by
    ($\vardiamond$). To conclude, we observe that both $\Gg_{\tilde \alpha}$
    and $\Hh_{\tilde \alpha}$ are perpendicular to the constant functions.
    \begin{claim}\label{cl:perp}
        For all $\tilde \alpha \in X\sqbr{[n] \setminus (I \cup J)}$ we have,
        \[ \inpr*{\Hh_{\tilde \alpha}, \one}_{\pi^{(\tilde \alpha)}_{I}} =
        \inpr*{\Gg_{\tilde \alpha}, \one}_{\pi^{(\tilde \alpha)}_{J}} = 0.\]
    \end{claim}
    \begin{claim}\label{cl:nloc}
        For any $S, T \subset [n]$ with $S \cap T = \varnothing$ and $\Ff:
        X\sqbr{S \cup T} \to \RR$, we have
        \[ \norm{\Ff}^2_{S \cup T} = \Exp_{\alpha \sim \pi_T}
        \norm{\Ff_\alpha}^2_{\pi^{(\alpha)}_S}.\]
    \end{claim}
    Thus, we can continue upper-bounding \cref{eq:chainlaststep} as follows,
    \begin{align*}
        \inpr*{\etch, \gee}_{\pi}
        &~\le~\Exp_{\tilde\alpha \sim \pi_{[n] \setminus (I \cup J)}}\sqbr*{
            \sigma_2\parens*{\Cee^{J \to I}_{\tilde \alpha}} \cdot
            \norm{\Hh_{\tilde \alpha}}_{\pi_{I}^{(\tilde \alpha)}}
        \cdot \norm{\Gg_{\tilde \alpha}}_{\pi_{J}^{(\tilde \alpha)}}},\\
        &~\le~\Exp_{\tilde\alpha \sim \pi_{[n] \setminus (I \cup J)}}\sqbr*{
            \sigma_2\parens*{ \Cee^{J \to I}_{\tilde \alpha} }\cdot
            \frac{\norm{\Hh_{\tilde \alpha}}_{\pi^{(\tilde \alpha)}_{I}}^2 +
        \norm{\Gg_{\tilde \alpha}}_{\pi^{(\tilde \alpha)}_{I}}^2}2},\\
        &~\le~\ee^{J \to I}.
    \end{align*}
    We now present the proofs of \cref{cl:perp} and \cref{cl:nloc}.
    \begin{proof}[Proof of \cref{cl:perp}]
        Let $\tilde \alpha \in X\sqbr{[n] \setminus (I \cup J)}$ be
        fixed. Since the proofs will be identical we will only prove
        the statement for
        $\Hh_{\tilde \alpha}$. Since $\etch$ is picked as an optimizer to
        \cref{eq:optf}, we have $\inpr{\etch, \uuu_{\tilde \alpha}}_\pi = 0$.
        We also have,
        \[ \inpr*{\uuu_{\tilde \alpha}, \etch}_{\pi} = \sum_{\alpha \in
            X\sqbr{[n] \setminus J}} \Hh(\alpha) \cdot
            \inpr*{\uuu_{\tilde \alpha},
            \uuu_\alpha}_{\pi} = \sum_{\alpha \in X\sqbr{[n]
            \setminus J}} \Hh(\alpha)
            \cdot \inpr*{\one, \uuu_{\alpha}}_\pi \cdot \one[\alpha
        \supset \tilde\alpha],\]
        where we have used that $\supp(\uuu_\alpha) \subset \supp(\uuu_{\tilde
        \alpha})$ whenever $\tilde \alpha \subset \alpha$.
        Thus, using $\inpr{\one, \uuu_\alpha}_{\pi} = \pi_{[n]
        \setminus J}(\alpha)$ we
        can simplify the above into
        \[ 0 = \sum_{\alpha \in X\sqbr{[n] \setminus J},\atop \alpha
                \supset \tilde
        \alpha} \Hh(\alpha) \cdot \pi_{[n] \setminus J}(\alpha).\]
        Every $\alpha \in X\sqbr{[n] \setminus J}$ satisfying $\alpha \supset
        \tilde\alpha$ can be written as $\alpha = \tilde \alpha \oplus
        \alpha_I$ for a unique choice of $\alpha_I \in
        X_{\tilde\alpha}\sqbr{I}$. Thus,
        using $\Hh(\tilde \alpha \oplus \alpha_I) = \Hh_{\tilde
        \alpha}(\alpha_I)$,
        we obtain
        \[ 0 = \sum_{\alpha_I \in X_{\tilde\alpha}\sqbr{I}} \Hh_{\tilde
            \alpha}(\alpha_I) \cdot \pi_{[n] \setminus J}(\tilde \alpha \oplus
            \alpha_I) = \pi_{[n] \setminus (I \cup J)}(\tilde \alpha)
            \cdot \sum_{\alpha_I \in
            X_{\tilde \alpha}\sqbr{I}} \Hh_{\tilde \alpha}(\alpha_I) \cdot
        \pi_{I}^{(\tilde \alpha)}(\alpha_I),\]
        where the last inequality is through \cref{eq:b-rule}. The theorem then
        follows by noting that the term on the RHS is a multiple of
        $\inpr*{\Hh_{\tilde
        \alpha}, \one}_{\pi_{I}^{(\tilde \alpha)}}$.
    \end{proof}

    \begin{proof}[Proof \cref{cl:nloc}]
        We have,
        \begin{align*}
            \norm{\Ff}_{\pi_{S \cup T}}^2 &~=~\sum_{\alpha \in X\sqbr{S \cup T}}
            \Ff(\alpha)^2 \cdot \pi_{S \cup T}(\alpha),\\
            &~=~\sum_{\alpha \in X\sqbr{S \cup T}}
            \pi_S(\alpha_S)
            \pi_T^{(\alpha_S)}(\alpha_T)\Ff_{\alpha_S}(\alpha_T)^2 ,\\
            &~=~\Exp_{\alpha_S \sim \pi_S}
            \norm{\Ff_{\alpha_S}}^2_{\pi^{(\alpha_S)}_T}
        \end{align*}
        where we have used \cref{eq:b-rule} to obtain the second equality.
    \end{proof}
\end{proof}

\section{An Entropy Contraction Inequality for $\Psq$}
\begin{theorem}\label{thm:ecc}
    Let $(X, \pi)$ be an $n$-partite simplicial complex and $\vec s = (s(1),
    \ldots, s(n))$ be any ordering of $[n]$. Then,
    \begin{equation}\EC(\Psq^{(\vec s)}) \ge \prod_{j = 1}^{n - 1} \eta^{s([j+1,
        n]) \to s(j)},\label{eq:mls-bd}
    \end{equation}
    where we have used the notation $s(A) = \set*{ s(a) : a \in A}$ and the
    parameters $\eta^{\bullet \to \bullet}$ are defined as in
    \cref{eq:dsplit}.
\end{theorem}
\begin{proof}[Proof of \cref{thm:ecc}]
    By renaming the parts, we assume $\vec s = (1, \ldots, n)$ and introduce the
    short-hand
    \[K(\ell, m) := K_{(X, \pi)}(\ell, m) = \prod_{j = \ell}^{m-1} \eta_{(X,
    \pi)}^{[j+1, m] \to j},\]
    and set $K(1, 1) = 1$.\footnote{Similarly, $K(\ell, \ell)
    = 1$ for all $\ell \le n$ but, we will not use this.} Recall the
    definition of the \ref{eq:mlsc} $\EC(\Psq)$,
    \[\EC(\Psq) =  1 - \kappa(\Psq)~~\textrm{where}~~\kappa(\Psq) = \sup\set*{
            \frac{\Div(\mu\Psq~\|~\pi)}{\Div(\mu~\|~\pi) } : \mu \in
    \triangle_{X^{(n)}}}.\]
    Thus, it suffices to prove $\kappa(\Psq) \le 1 - K(1, n)$. Equivalently,
    \[ \Div(\mu\Psq~\|~\pi) \le (1 - K(1,n)) \cdot \Div(\mu~\|~\pi)~~\textrm{for
    all}~~\mu \in \triangle_{X^{(n)}}.\]
    We will proceed by induction on the number of parts $n$ of the simplicial
    complex $(X, \pi)$.

    \textbf{Induction basis.} Suppose $n = 1$. We note, $\Psq =
    \Quu_1 = \one \pi$. In particular, we have
    \[ \Div(\mu\Psq~\|~\pi) = \Div(\mu \one\pi~\|~\pi) =
        \Div(\pi~\|~\pi) = 0 \le
    (1 - K(1,1)) \cdot \Div(\mu~\|~\pi).\]
    Since $K(1,1) = 1$, the induction hypothesis holds.

    \textbf{Induction step.} Let $m \ge 1$. Suppose that for any
    $n$-partite simplicial complex
    $(\hat X, \hat\pi)$ and all $\hat\mu \in \triangle_{\hat X^{(n)}}$, we have,
    \begin{equation}
        \Div( \hat\mu\widehat\Psq~\|~\hat\pi) \le (1 -K_{(\hat X,
        \hat\pi)}(1,n)) \cdot
        \Div(\hat\mu~\|~\hat\pi)\tag{induction hypothesis}\label{eq:entih}
    \end{equation}
    where we have denoted the sequential sweep walk in the complex $(\hat X,
    \hat \pi)$ by $\widehat{\Psq}$.

    Now, suppose $(\Omega, \pi)$ is an $(n+1)$-partite simplicial complex. By
    the chain-rule for \ref{eq:kl-def} (\cref{fac:cr}), we have
    \begin{equation}
        \Div( \mu\Psq~\|~\pi) = \Div( (\mu\Psq)_{\set{1}}~\|~\pi_{\set{1}}) +
        \Exp_{x \sim (\mu\Psq)_{\set{1}}}\Div( (\mu\Psq)^{(x)}_{[2,
        n+1]}~\|~(\pi)^{(x)}_{[2, n+1]})
    \end{equation}
    where $\nu_{\set 1}$ denotes the $\set 1$-marginal of the distribution $\nu
    \in \triangle_{X^{(n)}}$, qv.~\cref{eq:m-def}.

    We claim,
    \begin{claim}\label{cl:houd0}
        For all $x \in X[1]$, we have
        \begin{enumerate}
            \item $(\mu\Psq)_{\set 1} = (\mu\Quu_1)_{\set
                1}$,
            \item $(\mu\Psq)^{(x)}_{[2, m+1]} = (\mu \Quu_1)^{(x)} \Psq^{(x)}$,
        \end{enumerate}

        where we have written $\Psq^{(x)}$ or the sequential sweep
        in the
        $n$-partite
        complex $(X_x, \pi^{(x)}_{[2, n+1]})$.
    \end{claim}

    In particular, by the \ref{eq:entih} and the second item in \cref{cl:houd0},
    \begin{align}
        \Div( \mu\Psq~\|~\pi)
        &~=~\Div( (\mu\Psq)_{\set{1}}~\|~\pi_{\set{1}}) +
        (1 - K(2, n+1)) \Exp_{w \sim (\mu\Quu_1)_{\set{1}}}\Div(
            (\mu\Quu_1)^{(w)}_{[2,
        n+1]}\Psq^{(w)}~\|~(\pi)^{(w)}_{[2, n+1]}),\notag\\
        &~=~\Div( (\mu\Psq)_{\set 1}~\|~\pi_{\set 1})
        + (1- K(2,n+1))\parens*{\Div(\mu\Quu_1~\|~\pi) - \Div(
        (\mu\Quu_1)_{\set 1}~\|~\pi_{\set 1})},\label{eq:spag}
    \end{align}
    where we have used \cref{fac:cr} once again in obtaining
    \cref{eq:spag} and the
    $n$-sides of any $X_x$ for $x \in X[1]$, correspond to sides $2,
    \ldots, n+1$
    of $X$. Using the first item in \cref{cl:houd0} and
    \cref{eq:spag}, we obtain
    \begin{align}
        \Div( \mu\Psq~\|~\pi) &~\le~(1 - K(2,n+1)) \cdot
        \Div(\mu\Quu_1~\|~\pi) +
        K(2, n+ 1) \Div( (\mu\Quu_1)_{\set 1}~\|~\pi_{\set 1}),\notag\\
        &~\le~(1-K(2,n+1)) \cdot \Div(\mu~\|~\pi)
        + K(2,n+1)\cdot \Div(
        (\mu\Quu_1)~\|~\pi),\label{eq:hotd}
    \end{align}
    where we have used the data-processing inequality \cref{fac:dpi}
    in obtaining
    the last inequality.

    \begin{claim}\label{cl:houd2}
        We have,
        $(\mu\Quu_1)_{\set 1} = \mu_{[2, n+1]} \Cee_\varnothing^{[2, n+1] \to
        1}$.
    \end{claim}
    Thus, we have
    \begin{align}
        \Div( (\mu\Psq)~\|~\pi)
        &~\le~(1 - K(2, n+1)) \cdot
        \Div(\mu~\|~\pi) + K(2,n+1) \cdot \Div( \mu_{[2,
            n+1]}\Cee_\varnothing^{[2,
                n+1]
        \to 1}~\|~\pi_{\set 1}),\notag\\
        &~\le~(1 - K(2,n+1)) \Div(\mu~\|~\pi)
        + K(2,n+1)\cdot (1 - \eta^{[2, n+1] \to 1}) \cdot
        \Div(\mu_{[2, n+1]}~\|~\pi_{[2, n+1]}),\notag\\
        &~=~(1-K(2,n+1)) \Div(\mu~\|~\pi) + K(2, n+1) (1- \eta^{[2,
        m+ 1] \to 1}) \Div(\mu~\|~\pi),\label{eq:verylast}
    \end{align}
    where \cref{eq:verylast} is obtained via the chain-rule
    \cref{fac:cr} and the
    non-negativity of the \ref{eq:kl-def}. Simplifying, we obtain
    \[ \Div( \mu\Psq~\|~\pi) ~\le (1 - K(2, n+1)\cdot \eta^{[2, n+1] \to 1})
    \cdot \Div(\mu~\|~\pi).\]
    Now, the theorem statement follows by noticing $K(1, n+1) = K(2,
    n + 1) \cdot
    \eta^{[2, n + 1] \to 1}$. We conclude by proving \cref{cl:houd0}
    and \cref{cl:houd2}.
    \begin{proof}[Proof of \cref{cl:houd0}]
        Let $x \in X[1]$ be arbitrary. By \cref{eq:m-def},
        \begin{align*}
            [\mu\Psq]_{\set 1}(x)
            &~=~\sum_{\bar\omega \in X_x[2,\ldots, n+1]}
            [\mu \Quu_1 \cdots \Quu_{n+1}](x \oplus \bar \omega),\\
            &~=~\sum_{\bar\omega \in X_x[2, \ldots, n+1]}
            \sum_{\tilde \omega \in X_x[2, \ldots, n+1]}
            [\mu \Quu_1](x \oplus \tilde \omega) \cdot [\Quu_2 \cdots
            \Quu_{n+1}](x
            \oplus \tilde \omega, x \oplus\bar\omega).
        \end{align*}
        For obtaining the second equality we have used that the first coordinate
        is invariant under a random step according to $\Quu_2\cdots
        \Quu_{n+1}$. We
        continue,
        \begin{align*}
            [\mu\Psq]_{\set 1}(x)
            &~=~\sum_{\bar\omega \in X_x[2, \ldots, n+1]}
            \sum_{\tilde \omega \in X_x[2, \ldots, n+1]}
            [\mu \Quu_1](x \oplus \tilde \omega) \cdot [\Quu_2 \cdots
            \Quu_{n+1}](x
            \oplus \tilde \omega, x \oplus\bar\omega),\\
            &~=~\sum_{\tilde \omega \in X_x[2, \ldots, n+1]}[\mu
            \Quu_1](x \oplus \tilde \omega)
            \sum_{\bar\omega \in X_x[2, \ldots, n+1]}[\Quu_2 \cdots
            \Quu_{n+1}](x
            \oplus \tilde \omega, x \oplus\bar\omega)
        \end{align*}
        We note that since $\Quu_2\cdots \Quu_{n+1}$ does not change the first
        coordinate and all $\Quu_j$ are row-stochastic, the inner-sum
        equals 1. Thus,
        \[ [\mu\Psq]_{\set 1}(x)
            ~=~\sum_{\tilde \omega \in X_x[2, \ldots, n+1]}[\mu \Quu_1](x \oplus
            \tilde \omega)
        ~=~[(\mu\Quu_1)_{\set 1}](x),\]
        where we have used \cref{eq:m-def} and concluded the proof of the first
        item. For the second item, we again fix an
        arbitray $x \in X[1]$. Then, for any $\bar\omega \in
        X_x[2,\ldots, n+1]$ we
        have,
        \begin{align*}
            [\mu\Psq](x \oplus \bar\omega)
            &~=~ [\mu \Quu_1 \cdots \Quu_{n+1}](x \oplus \bar \omega),\\
            &~=~ \sum_{\tilde \omega \in X_x[2,\ldots, n+1]} [\mu
            \Quu_1](x \oplus \tilde \omega) \cdot [\Quu_2
            \cdots \Quu_{n+1}](x \oplus \tilde \omega, x \oplus \bar\omega).
        \end{align*}
        where we have used again that a random step according to $\Quu_2\cdots
        \Quu_{n+1}$ does not change the first coordinate in obtaining the second
        equality. Similarly, from \cref{eq:q-def} it is immediate
        that for all $j
        > 1$, one has $\Quu_j(x \oplus \omega',x \oplus \omega'')
        = \Quu_j^{(x)}(\omega', \omega'')$ for any pair $\omega, \omega'' \in
        X_x[2,\ldots, n+1]$ -- where $\Quu_j^{(x)}$ is the $j$-th
        update operator
        associated with $X_x$. Thus,
        \[ [\Quu_2 \cdots \Quu_n](x \oplus \tilde \omega, x \oplus \bar\omega) =
            [\Quu^{(x)}_2 \cdots \Quu^{(x)}_{n+1}](\tilde \omega, \bar\omega) =
        \Psq^{(x)}(\tilde \omega, \bar\omega).\]
        Computing further,
        \begin{align*}
            [\mu\Psq](x \oplus \bar\omega)
            &~=~\sum_{\tilde \omega \in X_x[2,\ldots, n+1]} [\mu\Quu_1](x \oplus
            \tilde \omega) \cdot \Psq^{(x)}(\tilde \omega, \bar\omega),\\
            &~=~\sum_{\tilde \omega \in X_x[2,\ldots, n+1]}
            [\mu \Quu_1]_{\set 1}(x) [\mu\Quu_1]^{(x)}(\tilde\omega) \cdot
            \Psq^{(x)}(\tilde \omega, \bar\omega).
        \end{align*}
        where we have used the Bayes' rule \cref{eq:b-rule}, in
        obtaining the last
        equality. Dividing both sides of the equation by
        $[\mu\Quu_1]_{\set 1}(x)$ we
        obtain,
        \[ \frac{[\mu\Psq](x \oplus \bar\omega)}{[\mu\Quu_1]_{\set 1}(x)} =
        [(\mu\Quu_1)^{(x)} \Psq^{(x)}](\bar\omega).\]
        Now, appealing to the first item $[\mu\Quu_1]_{\set 1} = [\mu\Psq]_{\set
        1}$ and the law of conditional probability \cref{eq:b-rule},
        the second statement
        follows.
    \end{proof}

    \begin{proof}[Proof of \cref{cl:houd2}]
        Let $x \in X[1]$ be arbitrary. We compute,
        \begin{align*}
            [\mu\Quu_1]_{\set 1}(x)
            &~=~\sum_{\bar \omega \in X_x[2, \ldots, n+1]}
            [\mu\Quu_1](x \oplus \bar \omega),\\
            &~=~\sum_{\bar \omega \in X_x[2, \ldots, n+1]} \sum_{y
            \in X_{\bar w}[1]}
            \mu(y \oplus \bar \omega)\pi_{\set 1}^{(\bar
            \omega)}(x),&& (\textrm{by }
            \cref{eq:q-def})\\
            &~=~\sum_{\bar \omega \in X_x[2, \ldots, n+1]}\pi_{\set 1}^{(\bar
            \omega)}(x)
            \cdot \sum_{y \in X_{\bar \omega}[1]} \mu(y \oplus \bar\omega),&&
            (\textrm{by } \cref{eq:m-def})\\
            &~=~\sum_{\bar w \in X_x[2, \ldots, n+1]} \mu_{[2,
            n+1]}(\bar \omega) \cdot
            \pi_{\set 1}^{(\bar \omega)}(x),&& (\textrm{by } \cref{eq:c-def})\\
            &~=~\sqbr*{\mu_{[2, n+1]} \Cee_\varnothing^{[2, n+1] \to 1}}(x).
        \end{align*}
    \end{proof}
\end{proof}

\section{The Sequential Sweep in Ramanujan Complexes}\label{sec:ramanujan}
In this section $X$ is a $(d-1)$-dimensional Ramanujan complex with $N$ vertices
associated with $\mathsf{PGL}_d(\mathbb{Q}_{p})$, as described in
Section \ref{sub:ramanujan_complexes}. This is a $d$-partite
simplicial complex with
sides $S_0, \ldots, S_{d-1}$. We use the $q$-number, $q$-factorial
and $q$-binomial notations:
\[
    [d]_{q}=\frac{q^{d}-1}{q-1},\qquad[d]!_{q}=[1]_{d}[2]_{d}\ldots[d]_{q},\qquad\left[
        \begin{smallmatrix}d\\
            j
    \end{smallmatrix}\right]_{q}=\frac{[d]!_{q}}{[j]!_{q}[d-j]!_{q}}.
\]
Since the vertex links of $X$ are isomorphic to the spherical complex
of $\mathsf{PGL}_d(\mathbb{F}_p)$, the number of facets in $X$
is
\[ \Abs*{X^{(d)}} = \frac{N \cdot [d]!_p}{d} \approx \frac{N
p^{\binom{d}2}}{d}.\]
First, we will use the structure of the spherical complexes to bound
the colored walks corresponding
to links $\sigma \in X[j+1,\ldots, d-1]$ for $j \ge 0$.
\begin{theorem}\label{thm:ramanujan_local}
    Let $1\leq j\leq d-2$, and let $\sigma\in X[j+1,\ldots,d-1]$ be a
    fixed face. Then
    \[
        \sigma_{2}\left(\Cee_{\sigma}^{\{0,\ldots,j-1\}\rightarrow\{j\}}\right)^{2}
        =\frac{p^{j+1}-p}{p^{j+2}+p^{j+1}-p-1} = \frac{1}{p+1} \cdot
        \parens*{1 - \Theta(p^{-j})}.
    \]
\end{theorem}
We present the proof of \cref{thm:ramanujan_local} in
\cref{sub:colored_walks_corresponding_to_non_trivial_links}. Next, by
using the spectral estimates for the Hecke-Operators defined in
\cref{sub:ramanujan_complexes}, we prove the following bound for the
colored walk in the empty link:
\begin{theorem}\label{thm:ramanujan_global}
    The second singular value of the operator
    $\Cee_{\varnothing}^{\{0,\ldots,d-2\}\rightarrow\{d-1\}}$
    on $X$ satisfies
    \[
        \sigma_{2}\left(\Cee_{\varnothing}^{\{0,\ldots,d-2\}\rightarrow\{d-1\}}\right)^2=\mathinner{\frac{\left[d\right]!_{p}+\left(d^{2}p^{d-1}-\left[d\right]_{p}\right)\left[d-2\right]!_{p}}{\left(p+1\right)\cdot\left[d\right]!_{p}}}
        = \frac{1}{p+1} \cdot (1+ O(p^{-d+2})).
    \]
\end{theorem}

Combining \cref{thm:ramanujan_global,thm:ramanujan_local}, we obtain
the following spectral gap bound for the sequential sweep on a
Ramanujan complex.
\begin{theorem}\label{thm:ram-mix}
    The second singular value of the sequential sweep on a Ramanujan
    complex of thickness $p$ and dimension $d$ satisfies,
    \begin{align*}
        \sigma_2\parens*{\Psq}^2 &~\le~1- \frac{ (p^d - 1)p (1 +
        O(p^{ -d +1}))}{(p-1) (p+1)^d},\\
        &~\le~\frac{d-1}p + O\parens*{\frac{d^2}{p^2}}
    \end{align*}
    where the latter inequality assumes $d \ll p$.
\end{theorem}

\begin{proof}
    We will appeal to \cref{thm:csv} and use the bounds from
    \cref{thm:ramanujan_local,thm:ramanujan_global}.

    First, we note that the \emph{local}
    contributions (from
    \cref{thm:ramanujan_local}) have a nice
    telescoping property. Namely,
    \begin{equation}\label{eq:local}
        \prod_{j = 2}^{d-2} \parens*{1 - \parens*{\ee^{[j-1] \to \set j}}^2}
        =
        \prod_{j=1}^{d-2}\left[1-{\frac{{p}^{j+1}-p}{{p}^{j+2}+{p}^{j+1}-p-1}}\right]=\mathinner{\frac{{p}^{d}-1}{\left(p-1\right)\left(p+1\right)^{d-1}}}.
    \end{equation}
    Then, combining this with \cref{thm:ramanujan_global}, we have
    \begin{align*}
        \sigma_2\parens*{\Psq}^2 &~\le~1 - \parens*{\frac{p^d - 1}{(p-1)
        \cdot (p+1)^{d-1}}} \cdot \parens*{ 1 - \frac{1}{p+1}
        \parens*{1 + O(p^{-d})}}\\
        &~=~1- \frac{ (p^d - 1)p (1 + O(p^{ -d +1}))}{(p-1) (p+1)^d}
    \end{align*}
    By expanding, $(p+1)^d = p^d (1 + d^{-1})^d = p^d + (1 + O(d/p) +
    O(p^{-2}))$
    \[ \sigma_2\parens*{\Psq}^2\le 1 - \frac{1}{1 - \frac{d-1}p + O(p^{-2})} \]
    Assuming $d \ll p$,
    \[ \sigma_2\parens*{\Psq}^2 \le \frac{d-1}p +
    O\parens*{\frac{d^2}{p^2}}.\qedhere\]
\end{proof}
\subsection{Colored Walks Corresponding to Non-Trivial
Links}\label{sub:colored_walks_corresponding_to_non_trivial_links} 
\begin{proof}[Proof of \cref{thm:ramanujan_local}]

    The link of a $\sigma\in X[j+1,\ldots,d-1]$ is isomorphic to the complex
    $\mathbb{P}^{j+1}(\mathbb{F}_{p})$ whose cells are the flags in
    $\mathbb{F}_{p}^{j+2}$ (see Section
    \ref{sub:ramanujan_complexes}). The colored walk
    $\Cee_{\sigma,X}^{\{0,\ldots,j-1\}\rightarrow\{j\}}$
    in $X_{\sigma}$ corresponds to the colored walk
    $\Cee_{\varnothing,\mathbb{P}^{j+1}(\mathbb{F}_{p})}^{\{1,\ldots,j\}\rightarrow\{j+1\}}$
    in $\mathbb{P}^{j+1}(\mathbb{F}_{p})$, which is the walk between flags
    of type $\{1,\ldots,j\}$ in $\mathbb{F}_{p}^{j+2}$, and hyperplanes in
    $\mathbb{F}_{p}^{j+2}$, where a flag $\mathcal{F}=\left\{
    V_{1}<\cdots<V_{j}\right\} $
    is connected to a hyperplane $\mathcal{H}$ if $V_{j}<\mathcal{H}$.

    Let us compare this with the walk between subspaces of dimensions $j$ and
    $j+1$ (hyperplanes) in $\mathbb{F}_{p}^{j+2}$, whose normalized
    incidence matrix is
    $\Cee_{\varnothing,\mathbb{P}^{j+1}\mathbb{F}_{p}}^{\{j\}\rightarrow\{j+1\}}$,
    which we will henceforth refer as $\Cee$. As $\mathcal{F}$ is
    connected to $\mathcal{H}$ if and only if
    $V_{j}< \mathcal{H}$,
    we obtain
    \[\Cee_{\varnothing,\mathbb{P}^{j+1}(\mathbb{F}_{p})}^{\{1,\ldots,j\}\rightarrow\{j+1\}}=\mathbf{1}_{m}\otimes
    \Cee\]
    where $m$ is the number of completions of $V_{j}$ to a flag
    of type
    $\{1,\ldots,j\}$.  We observe,
    \[
        \sigma_2\parens*{\Cee_{\varnothing,\mathbb{P}^{j+1}(\mathbb{F}_{p})}^{\{1,\ldots,j\}\rightarrow\{j+1\}}}^2
        = \lambda_2( \one_m^* \one_m \tensor \Cee^* \Cee) =
    \lambda_2(\Cee^*\Cee)\]
    where the last equality follows observing that $\one_m^* \one_m =
    1$  (note that
        $\mathbf{1}_m^{*}=\pi$
        is the uniform distribution by the definitions in
    \cref{prop:adjoint-defn,thm:spectral}).

    We observe that $\left[2\right]_{p} =p +1$,
    which is the number of lines in $\mathbb{F}_{p}^{2}$, is also the
    number of hyperplanes in $\mathbb{F}_{p}^{j+2}$ containing a fixed
    $j$-space, and that $\left[j+1\right]_{p}$ is the number of $j$-spaces
    contained in a fixed hyperplane, which shows that the graph
    $G_{\varnothing,\mathbb{P}^{j+1}(\mathbb{F}_{p})}^{\{j\}\rightarrow\{j+1\}}$
    is a $(\left[2\right]_{p},\left[j+1\right]_{p})$-biregular bipartite
    graph $(L, R, E)$ where $|L| = \binom{j+2}{2}_p$ and $|R| =
    [j+2]_p$. Using that
    any two different hyperplanes $\mathcal{H},\mathcal{H}'$ intersect
    in a $j$-space, we have
    \[
        \left[ \Cee^{*}\Cee\right ](\mathcal{H},\mathcal{H}')=
        \begin{cases}
            \frac{1}{\left[2\right]_{p}} & \mathcal{H}=\mathcal{H}'\\
            \frac{1}{\left[2\right]_{p}\left[j+1\right]_{p}} &
            \mathcal{H}\neq\mathcal{H}'
        \end{cases}.
    \]
    Thus,
    $\Cee^{*}\Cee=\frac{1}{\left[2\right]_{p}}\left(\frac{1}{\left[j+1\right]_{p}}\mathbb{1}+\left(1-\frac{1}{\left[j+1\right]_{p}}\right)\Ide\right)$
    where $\mathbb{1} \in \RR^{[j+2]_p \times [j+2]_p}$ is the
    all-one matrix, and we obtain by a
    straightforward calculation
    \[
        \lambda_{2}\left(\Cee^{*}\Cee\right)=\frac{1}{[2]_p} \cdot
        \parens*{1- \frac{1}{[j+1]_p}}=
        \frac{\left[j+1\right]_{p}-1}{\left[2\right]_{p}\left[j+1\right]_{p}}=\frac{p^{j+1}-p}{p^{j+2}+p^{j+1}-p-1}.\qedhere
    \]
\end{proof}

\subsection{Colored Walks in the Global
Link}\label{sub:colored_walks_in_the_global_link} 
\begin{proof}[Proof of \cref{thm:ramanujan_global}]
    Writing now
    $\Cee=\Cee_{\varnothing}^{\{0,\ldots,d-2\}\rightarrow\{d-1\}}$, we
    observe that $\Cee^{*}\Cee$ walks from a vertex of color $d-1$ to another
    such vertex, such that the two are contained in two facets
    which intersect
    in a common face of color $\{0,\ldots,d-2\}$. The graph
    $G_{\varnothing}^{\{0,\ldots,d-2\}\rightarrow\{d-1\}}$
    is $\left(\left[d\right]!_{p},p+1\right)$-biregular. Recall from
    \cite{LubotzkySV05, LubotzkySV05b} that the (nonnormalized)
    Hecke operator
    $\Aye_{j}$ walks from
    a vertex of color $t$ to all of its neighbors of color
    $t+j\in\nicefrac{\mathbb{Z}}{d}$.
    We now show that the two step walk satisfies
    \begin{align}
        \overbrace{\left[d\right]!_{p}\cdot\left(p+1\right)\Cee^{*}\Cee}^{:=
        \Quu} &
        =\left[d-2\right]!_{p}\Aye_{d-1}\Aye_{1}+\left(\left[d\right]!_{p}-\left[d\right]_{p}\left[d-2\right]!_{p}\right)
        \Ide\label{eq:q_def}\\
        &
        =\left[d-2\right]!_{p}\Aye_{d-1}\Aye_{1}+\left[d\right]!_{p}\left(1-\frac{1}{\left[d-1\right]_{p}}\right)\Ide\nonumber
    \end{align}
    Let $\Ess$ be operator which counts the number of walks from a
    $(d-1)$-color vertex to another (distinct) one, through a wall (a
    $(d-1)$-rank face) which shares a facet with each of them. The
    operator $\Quu$ is the nonnormalized version of $\Cee^*\Cee$,
    which is a combination of $\Ess$ and $\Ide$, the latter coming
    from backtracking. There are $[d]!_p$ facets containing each
    vertex, so this is the coefficient of $\Ide$. To compute the
    coefficient of $\Ess$, we use Hecke operators. The combination
    $\Aye_{d-1}\Aye_1=\Aye_1^\top\Aye_1$ takes each vertex once to
    every vertex across a wall from it, and $\deg \Aye_1=[d]_p$ times
    back to itself, so that $\Aye_{d-1}\Aye_1=\Ess+[d]_p\Ide$.
    Comparing degrees, we discover that the coefficient of $\Ess$ in
    $\Quu$ must be $[d-2]!_p$, obtaining \eqref{eq:q_def}.

    The Hecke operators $\Aye_{j}$ are all normal and commute
    with each other, and on a Ramanujan complex the nontrivial eigenvalue
    of $\Aye_{j}$ is bounded by ${d \choose
    j}p^{\frac{j\left(d-j\right)}{2}}$
    \cite{LubotzkySV05,LubotzkySV05b}, so that
    \[
        \lambda_{2}\left(\Quu\right)=\left[d-2\right]!_{p}d^{2}p^{d-1}+\left[d\right]!_{p}-\left[d\right]_{p}\left[d-2\right]!_{p}.
    \]
    Thus, the theorem follows by normalization and \cref{eq:q_def}.
\end{proof}


\bibliographystyle{alpha}
\bibliography{vedat}

\appendix
\section{$\ee$-spectral independence and $\ee$-product distributions}
\begin{xmpl}\label{xmpl:ce}
    For every integer $k > 2$ and every $n \ge 2$, there exists an
    $n$-partite simplicial complex $(X, \pi)$ such
    that $\pi$ is $(1/k, \ldots, 1/k)$-spectrally independent and $(1/k,
    \ldots, 1/k)$-product.
\end{xmpl}

\begin{proof}
    Let $k$ be given, we define $G := G_{n} = ([n],
    \set{\set{1,2}})$, i.e.~$G$ is the graph consisting of a single edge
    between 1 and 2 and $(n-2)$ isolated vertices. We let $(X, \pi)$ be the
    simplicial complex of proper $(k+1)$-colorings of $G$ equipped
    with the uniform
    measure $\pi$.

    We have $X^{(0)} = \set*{ (i, c) : i \in [n], c \in [k+1]}$, where $(i,
    c)$ represents the vertex $i$ being colored with the color $c$. We have,
    \[ \Inf_\varnothing( (i, c); (j, d) ) = 0~~\textrm{whenever}~~\set{i, j} \ne
    \set{1,2}.\]
    In particular, restricting $\Inf_\varnothing$ to the blocks of
    vertex-colors pairs $(i, c)$ for $i \in \set{1, 2}$ we have
    \[ \Inf_{\varnothing}|_{i \in \set{1, 2}} =
        \begin{pmatrix}
            0 & \parens*{\frac{1}{k} - \frac{1}{k+1}} \cdot \one\one^\top -
            \frac{1}{k}\cdot \Ide\\
            \parens*{\frac{1}{k} - \frac{1}{k+1}} \cdot \one\one^\top    -
            \frac{1}{k}\cdot \Ide & 0
    \end{pmatrix}\]
    It is simple to observe now, $\lambda_{\max}(\Inf_\varnothing) = 1/k$ and
    also we observe,
    \[ \Cee_\varnothing^{1 \to 2}(c_1, c_2) = \frac{\one[c_1 \ne
    c_2]}{k}~~\textrm{for all}~~c_1, c_2 \in [k+1]\]
    since this is equivalent to the random walk matrix of the complete graph,
    we have $\sigma_2(\Cee_\varnothing^{1 \to 2}) = 1lk$.
    Thus, $\pi$ is both $(1/k)$-spectrally independent and $(1/k)$-product --
    since whenever $\set{i, j} \ne \set{1, 2}$ we have
    $\sigma_2(\Cee_\varnothing^{i \to j}) = 0 $. The
    proof for the remaining levels follows analogously.
\end{proof}

\end{document}